%% file: main.tex
\documentclass[runningheads]{llncs}

\usepackage{hyperref}
\usepackage{xcolor}

\usepackage{graphicx}
\usepackage{eso-pic}
\usepackage{algorithm}
\usepackage[noend]{algpseudocode}
\usepackage{array, makecell} 
\usepackage{multirow}
\usepackage{epstopdf}
\usepackage{ulem}
\usepackage{caption}
\usepackage{footnote}
\usepackage{booktabs}
\usepackage{amsmath}
\usepackage{amsfonts}
\usepackage{amssymb}

\newcommand{\ignore}[1]{}

\newtheorem{observation}[theorem]{Observation}

\title{FairTraDEX: A Decentralised Exchange Preventing Value Extraction}

\usepackage{authblk}
\author{Conor McMenamin\inst{1,2} \and Vanesa Daza\inst{1,3} \and Matthias Fitzi\inst{4} \and Padraic O'Donoghue}

\authorrunning{C. McMenamin et al.}

\institute{Department of Information and Communication Technologies, Universitat Pompeu Fabra, Barcelona, Spain \and
NOKIA Bell Labs, Nozay, France \and CYBERCAT - Center for Cybersecurity Research of Catalonia \and IOHK }

\newcommand\nnfootnote[1]{%
  \begin{NoHyper}
  \renewcommand\thefootnote{}\footnote{#1}%
  \addtocounter{footnote}{-1}%
  \end{NoHyper}
}

\input{Variables}

\input{VariablesAlgorithm}

\begin{document}

\maketitle

\nnfootnote{ email | (primary) \textit{conor.mcmenamin@upf.edu}  \\
\textit{vanesa.daza@upf.edu} | \textit{matthias.fitzi@iohk.io} | \textit{ padraicodonoghue@gmail.com}}
\nnfootnote{\begin{minipage}{0.06\textwidth}
    \includegraphics[width=\linewidth]{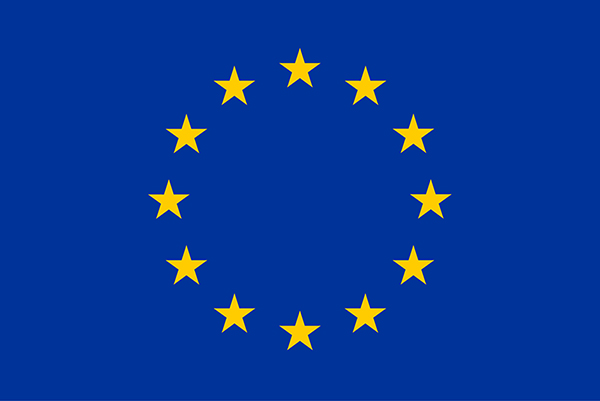}
    \end{minipage}%
    \hfill%
    \begin{minipage}{0.88\textwidth} This Technical Report is part of a project that has received funding from the European Union's Horizon 2020 research and innovation programme under grant agreement number 814284 \end{minipage}} 

\input{Abstract}


\input{Introduction/Introduction}

\input{RelatedWork}

\input{Prelims/Preliminaries}

\input{WSFBAs/WSFBAs}
\input{FairTraDEX/FairTraDEX}

\input{Attacks}

\input{Conclusion}

\addcontentsline{toc}{section}{Bibliography}
\bibliographystyle{splncs04}
\bibliography{references}

\appendix

\input{Appendices/Appendices}

\end{document}

%% file: Variables.tex
\usepackage{cancel}

\usepackage{pifont}

\newcommand{\any}{\textit{any}}
\newcommand{\auctionName}{\text{FBA}}
\newcommand{\balanceTokens}{\textit{bal}}
\newcommand{\bid}{\textit{bid}}

\def\bitcoin{%
  \leavevmode
  \vtop{\offinterlineskip 
    \setbox0=\hbox{B}%
    \setbox2=\hbox to\wd0{\hfil\hskip-.03em
    \vrule height .3ex width .15ex\hskip .08em
    \vrule height .3ex width .15ex\hfil}
    \vbox{\copy2\box0}\box2}}
\newcommand{\blacklistedSerialNums}{\textit{blacklistedSNs}}
\newcommand{\blockchain}{\textit{Blockchain}}

\newcommand{\buyVolume}{\textit{buyVolume}}
\newcommand{\buyVolumeNew}{\textit{buyVolumeNew}}

\newcommand{\clearingPrice}{\textit{CP}}

\newcommand{\clientCommits}{\textit{clientCommits}}
\newcommand{\clients}{\textit{clients}}
\newcommand{\commit}{h}
\newcommand{\commitment}{\textit{com}}
\newcommand{\commitments}{\textit{Com}}
\newcommand{\commRandomness}{\textit{r}}
\newcommand{\commSerialNum}{\textit{S}}

\newcommand{\commitDeadline}{\textit{requestDeadline}}

\newcommand{\cryptoParam}{\kappa}
\newcommand{\currentAuctionNotional}{\textit{currAucNotional}}

\newcommand{\escrow}{\textit{escrow}}
\newcommand{\escrowClient}{\escrow_\textit{client}}
\newcommand{\escrowMM}{\escrow_\textit{MM}}

\newcommand{\escrowMAX}{Y}
\newcommand{\fee}{\textit{f}_{mcf}}
\newcommand{\feeRelayer}{\textit{f}_r}

\newcommand{\getAppend}{\textit{.append}}
\newcommand{\getBalance}{\textit{.balance}}

\newcommand{\getHeight}{\textit{.height}}

\newcommand{\getTransfer}{\textit{.transfer}}
\newcommand{\getWidth}{\textit{.width}}
\newcommand{\height}{H}
\newcommand{\imbalance}{\textit{imbalance}}

\newcommand{\lastPhaseChange}{\textit{lastPhaseChange}}

\newcommand{\market}{\textit{market}}
\newcommand{\marketOrder}{\textit{mkt}}

\newcommand{\maxAuctionNotional}{\textit{Q}_\textit{not}}

\newcommand{\MIFP}{y}
\newcommand{\MIFPA}{\MIFP_A}
\newcommand{\MIFPB}{\MIFP_B}
\newcommand{\MIFPAB}{\MIFP_{A \rightarrow B}}

\newcommand{\minTickSize}{\textit{minTickSize}}
\newcommand{\MM}{\textit{MM}}
\newcommand{\MMBidPrice}{\bid}
\newcommand{\MMBidAmount}{\tokenAmount_\bid}
\newcommand{\MMOfferPrice}{\offer}
\newcommand{\MMOfferAmount}{\tokenAmount_\offer}

\newcommand{\MMCommits}{\textit{MMCommits}}

\newcommand{\negligible}{\textit{negl}}
\newcommand{\netTradeSize}{\textit{D}}

\newcommand{\notional}{X\bitcoin}
\newcommand{\numMMs}{N}

\newcommand{\offer}{\textit{offer}}

\newcommand{\order}{\textit{order}}
\newcommand{\ourAuctionName}{\text{WSFBA}}
\newcommand{\params}{\textit{params}}
\newcommand{\phase}{\textit{phase}}
\newcommand{\phaseReveal}{\textit{Reveal}}
\newcommand{\phaseCommit}{\textit{Commit}}
\newcommand{\phaseResolution}{\textit{Resolution}}

\newcommand{\player}{\textit{P}}
\newcommand{\playerProtocol}{\textit{player}}
\newcommand{\playeri}{\player_i}
\newcommand{\players}{\textit{players}}
\newcommand{\postTradeImpact}{\delta}
\newcommand{\priceToCheck}{\textit{priceToCheck}}
\newcommand{\proofZK}{\pi}
\newcommand{\proRatePrice}{\tokenPrice_\textit{pro-rate}}
\newcommand{\proRatePriceVolume}{\tokenAmount_\textit{pro-rate}}
\newcommand{\protocolContract}{\textit{protocolContract}}
\newcommand{\protocolName}{\text{FairTraDEX}}
\newcommand{\refPrice}{\MIFP_\textit{ref}}
\newcommand{\regToken}{\textit{regID}}
\newcommand{\regTokenNew}{\regToken\textit{New}}
\newcommand{\regTokens}{\textit{RegIDs}}
\newcommand{\relay}{\textit{relay}}
\newcommand{\remove}{\textit{.remove}}
\newcommand{\resBounty}{\textit{resBounty}}
\newcommand{\revealDeadline}{\textit{revealDeadline}}
\newcommand{\revealedMkts}{\textit{revealedMkts}}
\newcommand{\revealedOrders}{\textit{revealedOrders}}
\newcommand{\revealedBuyOrders}{\textit{revealedBuyOrders}}
\newcommand{\revealedSellOrders}{\textit{revealedSellOrders}}
\newcommand{\revealTXDelay}{T}
\newcommand{\secParam}{\psi}
\newcommand{\sellVolume}{\textit{sellVolume}}
\newcommand{\sellVolumeNew}{\textit{sellVolumeNew}}
\newcommand{\strategy}{\textit{str}}

\newcommand{\tagClientCommit}{\textit{COMMIT}}
\newcommand{\tagClientRegister}{\textit{CLIENT-REGISTER}}
\newcommand{\tagClientReveal}{\textit{CLIENT-REVEAL}}
\newcommand{\tagMMCommit}{\textit{COMMIT}}
\newcommand{\tagMMReveal}{\textit{MM-REVEAL}}

\newcommand{\tieBreaker}{\textit{tieBreaker}}
\newcommand{\tieBreakSeed}{\textit{tieBreakSeed}}
\newcommand{\tightestMarket}{\textit{tightMkt}}

\newcommand{\tightestWidth}{\width_\textit{tight}}

\newcommand{\token}{\textit{tkn}}
\newcommand{\tokenA}{\textit{A}_\textit{tkn}}

\newcommand{\tokenAmount}{\textit{size}}
\newcommand{\tokenPrice}{\textit{p}}
\newcommand{\tokenTradeSize}{\textit{tokenTradeSize}}
\newcommand{\tokenB}{\textit{B}_\textit{tkn}}

\newcommand{\volumeSettled}{\textit{volumeSettled}}
\newcommand{\withdraw}{\textit{withdraw}}
\newcommand{\width}{\textit{w}}

%% file: VariablesAlgorithm.tex
\algdef{SE}[EVENT]{Upon}{EndUpon}[1]{\textbf{upon}\ #1\ \algorithmicdo}{\algorithmicend}\algtext*{EndUpon}

\newcommand{\newconstruct}[5]{%
  \newenvironment{ALC@\string#1}{\begin{ALC@g}}{\end{ALC@g}}
   \newcommand{#1}[2][default]{\ALC@it#2\ ##2\ #3%
     \ALC@com{##1}\begin{ALC@\string#1}}
   \ifthenelse{\boolean{ALC@noend}}{
     \newcommand{#4}{\end{ALC@\string#1}}
   }{
     \newcommand{#4}{\end{ALC@\string#1}\ALC@it#5}
   } 
}

\newcommand{\assign}{\leftarrow}

\newcommand{\ident}[1]{\mathit{#1}}
\def\newident#1{\expandafter\def\csname #1\endcsname{\ident{#1}}}

\newcommand{\with}{\textbf{with}}

\newcommand{\from}{\textbf{from}}

\newcommand{\SPACE}{\vspace{3mm}}

\newcommand{\logicalAnd}{\wedge}
\newcommand{\logicalOr}{\vee}
\newcommand{\logicalNot}{\neg}

%% file: Abstract.tex
\begin{abstract}
  \ignore{An idealised decentralised-exchange (DEX) provides a medium in which players wishing to exchange one token for another can interact with other such players and liquidity providers at a price which reflects the true exchange rate, without the need for a trusted third-party. Unfortunately, extractable value is an inherent flaw in existing blockchain-based DEX implementations. This extractable value takes the form of monetizable opportunities that allow blockchain participants to extract money from a DEX without adding demand or liquidity to the DEX, the two functions for which DEXs are intended. This money is taken directly from the intended DEX participants. As a result, the cost of participation in existing DEXs is much larger than the upfront fees required to post a transaction on a blockchain and/or into a smart contract. }
  
  We present FairTraDEX, a decentralized exchange (DEX) protocol based on frequent batch auctions (FBAs), which provides formal game-theoretic guarantees against extractable value. FBAs when run by a trusted third-party provide unique game-theoretic optimal strategies which ensure players are shown prices equal to the liquidity provider's fair price, excluding explicit, pre-determined fees. FairTraDEX replicates the key features of an FBA that provide these game-theoretic guarantees using a combination of set-membership in zero-knowledge protocols and an escrow-enforced commit-reveal protocol. We extend the results of FBAs to handle monopolistic and/or malicious liquidity providers. We provide real-world examples that demonstrate that the costs of executing orders in existing academic and industry-standard protocols become prohibitive as order size increases due to basic value extraction techniques, popularized as maximal extractable value. We further demonstrate that FairTraDEX protects against these execution costs, guaranteeing a fixed fee model independent of order size, the first guarantee of it's kind for a DEX protocol. We also provide detailed Solidity and pseudo-code implementations of FairTraDEX, making FairTraDEX a novel and practical contribution.
\end{abstract}

%% file: Introduction/Introduction.tex
\section{Introduction}\label{sec:Introduction}

One of the most prominent and widely-used classes of protocols being run on smart-contract enabled blockchains is that of decentralised-exchange (DEX) protocols. DEX protocols allow a specific set of players, whom we call \textit{clients}, to exchange one token for another in the presence of \textit{market-makers} (MMs), who provide liquidity to clients, usually in exchange for a fee. Interacting with a blockchain-based DEX requires a client or MM to first interact with the players who add transactions to the blockchain, known as miners or block producers. These interactions typically reveal a player's intention to trade to the block producer before the transaction is confirmed on the blockchain, and in doing so, present the block producer with what has become known as a miner-extractable value (MEV) opportunity. MEV, first coined in \cite{FlashBoys2.0}, refers to any expected profits the miner of a block can extract from other players interacting with the blockchain. This extraction is performed by manipulating the ordering of, injecting, and/or censoring transactions in prospective blocks. This has been generalised to expected extractable value (EEV) \cite{EstimatingMEVLetsGoShoppingJudmayer} (defined in Appendix \ref{app:terminology}), as non-miner players can also perform many of these attacks.

\ignore{However, MEV is merely the tip of the iceberg. In blockchain-based protocols such as DEXs, it is also possible for one player to extract value from another by disobeying the protocol instructions, an attack we refer to as \textit{selective participation}. Many decentralised protocols do not bind players to correctly performing protocol actions. In such protocols, players are given a form of \textit{optionality}, the option to correctly perform protocol actions or not, which is typically not intended by the protocol. Optionality is well understood in derivative markets, and is something which must be charged for. Otherwise, players receive this optionality/value at a discount/for free. Given optionality exists in a protocol, this value can be used by one player, and taken from others. As such, value can theoretically be extracted by every player in the blockchain ecosystem, not just the miners. We refer to any value which can be extracted from players in a blockchain-based protocol other than fixed upfront fees as \textit{expected extractable value} (EEV) in line with recent work \cite{EstimatingMEVLetsGoShoppingJudmayer},  on the generalisation of extractable value. }

A significant advancement in DEX protocols was the advent of \textit{automated market makers} (AMMs), with  Uniswap \cite{UniswapWebsite} being the most prominent of which.\ignore{ based on the principle of maintaining a constant function between two tokens in a pool when depositing and swapping, forming the basis for the relative price of those tokens. AMMs are built to reflect the market-implied fair price of one token expressed in another, and to provide a permanent source of liquidity for the token swap, accessible by any player in the blockchain ecosystem. Over large time horizons, both of these purposes have been fulfilled. However, on a transaction-by-transaction level view, p} Projects like Flashbots \cite{FlashbotsExploreWebsite} (a direct spin-off to \cite{FlashBoys2.0}) have identified that AMMs are the main source of recorded EEV ($>98\%$, as seen in the chart labelled ``Extracted MEV Split by Protocol" in \cite{FlashbotsExploreWebsite}, of the $\$665$M in EEV identified by Flashbots since August, 2020). Furthermore, Flashbots only observes basic forms of EEV, meaning in reality (and as stated by the Flashbots team \cite{FlashbotsExploreFAQWebsite}), this amount of EEV is a lower bound for the total amount of value being extracted from clients and MMs alike through participation in DEX protocols. To this point,  a peer-reviewed analysis in \cite{QuantifyingExtractableValueGervais} identified $\$540.54$M in extracted value up to August 2021, indicating the current number provided by Flashbots is significantly lower than the actual amount of extracted value. In \cite{ImpermanentLossAMMsLoesch}, it has been further highlighted that in Uniswap V3, liquidity providers are losing more to EEV attacks (impermanent loss in this case) than they are collecting in fees. It is clear that the long-term viability of existing DEX protocols is not plausible.

Although many attempts have been made academically to address this significant source of EEV \cite{FairMMCiampi,P2DEXBaum,BlockAuctionPeriodicAuctionsConstantinides,AequitasKelkar,PubliclyVerifiableSecrecyPreservingPeriodicAuctionsGalal,HelixAsayag}, no satisfactory solution has been found. The protocols presented in these works remain vulnerable to basic EEV attacks in the case where all transactions are eventually added to the blockchain (censorship-resistant). The overriding trend in these papers is a purely cryptographic approach to an economic/game-theoretic problem, resulting in relatively straightforward game-theoretic exploits. When describing these protocols in Section \ref{sec:RelatedWork}, we also describe some of the value extraction attacks to which the respective protocols are vulnerable.
\ignore{For example, in \cite{P2DEXBaum,BlockAuctionPeriodicAuctionsConstantinides,PubliclyVerifiableSecrecyPreservingPeriodicAuctionsGalal} a requirement is for orders not to be revealed if they are not executed to preserve one element of user and order privacy. All the while, both the identity of all players entering the execution environment, and their intended direction, are revealed before being executed, key pieces of information which allow for value extraction.} 
Therefore, there is a clear gap, both in literature and in practice, to provide a DEX protocol which definitively eliminates all sources of EEV. In this paper, we provide such a protocol. We outline $\protocolName$, a DEX protocol based on an existing auction process called \textit{frequent batch auctions} (FBAs) \cite{FrequentBatchAuctionsBudish} and zero-knowledge (ZK) tools for set-membership. FBAs have been proven to provide a strict Nash Equilibrium in which clients trade at the market fair price for a particular token swap, although are intended to be run by a trusted third party. $\protocolName$ is, to the best of our knowledge, the first DEX protocol outlining specific practical conditions under which EEV is prevented in a censorship-resistant blockchain by replicating the guarantees of an FBA without the need for a trusted third party. We formally prove this security against EEV by utilising results from ZK literature and game-theory. 

\ignore{Specifically, we show that given a sufficient number\footnote{This number of participants is described in Section \ref{sec:propertiesOfFairTraDEX}, and used analogously to Zerocash\cite{ZCash}, and it's derivatives. An increase in minted tokens (registrations in $\protocolName$) decreasing the probability that a randomly selected token corresponds to the player who minted it.} of registered participants in $\protocolName$, the strategy of players trading all client tokens (MMs provide liquidity greater than client volume in equilibrium) at the market-implied fair price, minus fees, forms a strict Nash equilibrium. This is compared to DEX protocols in which EEV opportunities exist, where for example, a potential value-extractor responds differently to buy orders vs. sell orders, or once-off clients vs. professional MMs. $\protocolName$ replicates the private submission of orders to a trusted third-party in an FBA, revealing no public information to the blockchain protocol other than the set of all registered players, and the number of players in an auction. Many protocols in-use today\footnote{ZCash \cite{ZerocoinGreen}, Tornado Cash \url{https://tornado.cash/}, and Semaphore \url{https://semaphore.appliedzkp.org/}, to name but a few.} replicate this exact proving of set-membership in zero-knowledge functionality. These, coupled with our proof-of-concept Solidity implementation \cite{FairTraDEXGithub} provide a clear guide for the adoption of $\protocolName $ on every major smart-contract enabled blockchain.}

\input{Introduction/OurContribution}

\ignore{\input{Introduction/Organisation}}

%% file: Introduction/OurContribution.tex
\vspace{-0.2cm}
\subsection{Our Contribution}

We first introduce a width-sensitive frequent batch auction ($\ourAuctionName$), an idealised commit-reveal exchange protocol between clients and MMs, based on FBAs \cite{FrequentBatchAuctionsBudish}. $\ourAuctionName$s are an important improvement on basic FBAs with respect to decentralised systems. $\ourAuctionName$s ensure clients submit market orders in the presence of monopolistic MMs, compared to a standard FBA in which clients are required to submit limit orders at the fair price plus some trade fee. The requirement for clients to submit limit orders leads to worse order execution as trade probability is decreased, while also placing a significant burden on clients to choose the fair price. This burden is removed in an $\ourAuctionName$, providing an ``obvious optimal'' for clients, as coined in \cite{EIP1559roughgarden}. Furthermore, in the case of competing MMs, a $\ourAuctionName$ provides equivalent equilibria to \cite{FrequentBatchAuctionsBudish}. Importantly, this equilibrium involves all client orders being traded in the auction in which the orders are submitted, removing any strategies which involve unfilled orders needing to being revealed and resubmitted in proceeding auctions.

We then describe $\protocolName$, a blockchain-based implementation of a WSFBA.
\ignore{At a high-level, clients must first register to participate in the protocol by depositing an \textit{escrow} to the smart-contract. To then enter a $\ourAuctionName$ in the $\protocolName$ protocol, clients commit to an order along with a ZK proof proving that they deposited an escrow to the smart contract. MMs simultaneously commit to markets (bids and offers), although depositing an escrow at time of commitment (MM escrows do not reveal direction due their neutral directional nature). In the proceeding reveal phase, to ensure the correct revelation of information as in the reveal phase of $\protocolName$, any clients or MMs who committed to orders/markets in the commit phase but do not reveal lose access to their deposit. After the reveal phase finishes, the auction is settled at a single clearing price which maximises trade volume. A representation of the information and escrow flow in $\protocolName$ prior to order settlement is presented in Figure \ref{fig:protocolFlow}. }
In $\protocolName$, order commitments are recorded on-chain (to enforce the corresponding escrow punishment). We utilise ZK set-membership proofs to allow clients to commit to their orders anonymously. As such, in $\protocolName$, every client must initially register to the protocol, depositing an escrow. Then, whenever a client wants to commit to an order, the client only has to prove that they are a member of the set of players who registered in the protocol. Given enough registrations, the probability a client's ZK set-membership proof and committed order relates to the actual order contents approaches 0 (we formalise this notion in Section \ref{sec:FairTraDEX}) such that no other player in the system can see the committed order and use it to infer anything about what the order is. To definitively hide a client's order-information, orders are committed, including the ZK-proof, by using a relayer, a third-party who receives a fee for including relayed transactions in the blockchain (see Appendix \ref{sec:Relayers} for further details).

\ignore{This combination of tools serves as a subtle, yet crucial, improvement to previous attempts to implement blockchain-based FBAs \cite{P2DEXBaum,PubliclyVerifiableSecrecyPreservingPeriodicAuctionsGalal,BlockAuctionPeriodicAuctionsConstantinides} where game-theoretic guarantees are dependent on the hiding of order information until all markets and orders have been committed. In these previous attempts, on-chain order commitment reveals client/order-specific information which allows observant counterparties to improve their ability to guess client order information. With sufficiently many registered clients, no counterparty using $\protocolName$ can learn information that can be used to bias prices against clients with positive expectancy. As such, EEV is prevented in $\protocolName$ when enough clients register to the protocol. }

We provide an extensive Ethereum virtual-machine (EVM) compatible proof-of-concept for $\protocolName$ \cite{FairTraDEXGithubPublic} including a comparison of protocol running costs with previous solutions in Section \ref{sec:Attacks}, which remain constant with respect to order size, price and direction. When compared to potentially percentage-point slippages and EEV-attack costs required to trade on current DEXs, also highlighted in Section \ref{sec:Attacks}, $\protocolName$'s formal guarantees of protocol-level EEV prevention and up-front, fixed and explicit costs set a new standard for DEX protocols.

%% file: Introduction/Organisation.tex
\subsection{Organisation of the Paper}

Section \ref{sec:RelatedWork} analyses all previous work related to the construction of EEV-proof DEX protocols. Section \ref{sec:Prelims} outlines the financial terminology used in the paper, the player definitions (\ref{sec:financeTerms}), the ZK primitives (\ref{sec:ZK}) and relayer functionalities (\ref{sec:Relayers}) needed to formally reason about $\protocolName$. Section \ref{sec:FBAs} defines the ideal $\ourAuctionName$ functionality, and identifies the strategies of both clients and MMs. Section \ref{sec:FairTraDEX} maps the ideal $\ourAuctionName$ functionality to a series of algorithms which form the $\protocolName$ protocol. It is then proved that rational players follow $\protocolName$, and as such, implement a $\ourAuctionName$. Section \ref{sec:Encoding} provides a description and gas-cost analysis of the smart-contract encoding of $\protocolName$, including a comparison to existing solutions. We conclude in Section \ref{sec:conclusion}. The Appendices include an extended related work (App. \ref{app:RelatedWork}), extended discussion on ZK literature (App. \ref{app:ZK}), proofs not included in the main body of the paper (App. \ref{app:proofs}), a detailed description of the $\protocolName$ algorithms (App. \ref{app:FairTraDEXAlgos}), a description of a clearing price verifier for settling orders in $\protocolName$ (App. \ref{app:clearingPriceCalc}), the pseudo-code for encoding $\protocolName$ as a smart-contract (App. \ref{app:Protocol}), and notes on $\protocolName$ (App. \ref{app:discussion}).

%% file: RelatedWork.tex
\section{Related Work}\label{sec:RelatedWork}

The main works aimed at protecting DEX users from EEV either focus on preventing front-running of orders \cite{FairMMCiampi,P2DEXBaum}, the fair ordering of transactions based on their delivery time \cite{AequitasKelkar}, or on hiding client trade information until the trade has been committed to the blockchain \cite{PubliclyVerifiableSecrecyPreservingPeriodicAuctionsGalal,HelixAsayag,BlockAuctionPeriodicAuctionsConstantinides}. Of these works, the closest to our proposal are \cite{FairMMCiampi,P2DEXBaum,PubliclyVerifiableSecrecyPreservingPeriodicAuctionsGalal,BlockAuctionPeriodicAuctionsConstantinides}. All of these works critically depend on honesty from MMs, auction operators and/or the block proposers. In Appendix \ref{app:RelatedWork}, we provide high-level descriptions of the DEX protocols while in this section, we briefly outline how, when all players in the respective DEX protocols are rational, EEV opportunities exist. 

In \cite{FairMMCiampi}, the MM is always allowed to see orders from clients and can choose to abort them. It is argued that MMs are happy to trade against all orders, including informed orders. Furthermore, it is assumed that clients are independent, with random information. This is not true in real-world trading environments, and as such, Theorem 3 may not hold in practice. 
\ignore{The paper references \cite{OptimalSequntialMMingDas} as justification for the quality of price/service provided by the MM, however there is a subtle but crucial difference between the games in \cite{OptimalSequntialMMingDas}: in \cite{OptimalSequntialMMingDas} the client has the final decision on whether or not to execute the trade, while in \cite{FairMMCiampi} the MM has the final decision. This optionality has an implicit, but not explicit, cost for the client and provides a source of EEV to the MM. In $\protocolName$, no optionality is given to the MM or clients, while given enough registered clients, the direction of any individual client remains hidden until the trade has been committed. As such, $\protocolName$ is able to benefit from the results of \cite{OptimalSequntialMMingDas} in a single MM game, that is, liquid (tradeable and of a client-specified width) markets centred around the pre-trade market-implied fair value (MIFP) when a trade is accepted by the MM. }

In P2DEX \cite{P2DEXBaum}, clients must publicly deposit the tokens with which to trade in the same time frame as the order matching takes place, exposing clients to standard identity- and directional-based EEV exploits. Separating token deposit and identity revelation from a client's commitment to a specific auction are important advancements used in $\protocolName$ to protect against EEV. Furthermore, at least one of the servers in charge of fairly executing orders is required to be honest-by-default. If the servers, a minority subset of players in the P2DEX protocol, are rational and monopolised/colluding, the servers can front-run orders. 
In $\protocolName$, such value-extraction is prevented by keeping all order information hidden until every order in a particular settlement round has been committed. Furthermore, the set of servers act as a point-of-failure as server participation is required to finalise order-settlement. No subset of players in $\protocolName$ can prevent the matching of correctly-revealed orders, while in P2DEX, any majority of servers can prevent order-matching. 

Similar commit-reveal protocols to $\protocolName$ for blockchain-based token-exchange are proposed in \cite{PubliclyVerifiableSecrecyPreservingPeriodicAuctionsGalal,BlockAuctionPeriodicAuctionsConstantinides}. The protocol in \cite{PubliclyVerifiableSecrecyPreservingPeriodicAuctionsGalal} is exposed to several game-theoretic exploits which contradict its protection against front-running. These include the necessity to reveal order direction a-priori, and the non-trivial handling of the linkability between commitments and account-balances. The protocol also depends on an operator who does not participate in token-exchange, gains exclusive access to order information, and is depended on for protocol completion. 
In \cite{BlockAuctionPeriodicAuctionsConstantinides}, clients commit their own orders to the blockchain, revealing their identities, and corresponding token balances/execution patterns which can be used by a basic professional MM to skew prices and extract value from the client. Both of \cite{PubliclyVerifiableSecrecyPreservingPeriodicAuctionsGalal,BlockAuctionPeriodicAuctionsConstantinides} aim to protect unmatched orders from being revealed, but to do so, both protocols depend on a third party selected before the auction begins to execute order-matching (no one else in the ecosystem can finalise the auction, a single point of failure). Both protocols assume that the operator does not reveal unexecuted order information.

%% file: Prelims/Preliminaries.tex
\vspace{-0.3cm}

\section{Preliminaries}\label{sec:Prelims}

\input{Prelims/FinanceTerms}

\input{Prelims/ZKMachinery2}

%% file: Prelims/FinanceTerms.tex
This section introduces the terminology and definitions necessary to understand the main results of the paper. By $\textit{negl}()$ we denote any function $f:\mathbb{N}\rightarrow\mathbb{R}$ that decreases faster than any (positive) polynomial $p$. More formally, $\forall \ p  \ \exists \ \lambda_0\in\mathbb{N}: \forall \lambda>\lambda_0: f(\lambda)<\frac{1}{p(\lambda)}$.
For protocol correctness, we must assume that some of the involved players may be malicious trying to force the protocol into incorrect execution, and without any direct benefit for themselves. However, for the game-theoretic part of the analysis, we assume that all players are rational.
Accordingly, the analysis of our protocol is based on two security parameters, a cryptographic security parameter $\cryptoParam$ used to bound the probability that the protocol execution is incorrect, and a game-theoretic extractable-value parameter $\secParam$ used to bound the extractable value of any player.
The following are the crucial terms and definitions needed for the remainder of the paper, with supplementary terminology and definitions contained in Appendix \ref{app:terminology}.

\begin{itemize}
     \item \textit{Notional Value}: The value of a set of tokens expressed in some common reference token. In this paper, we use the symbol $\bitcoin$ as the reference token in which we measure notional, and with which we reason about utility. 
          
     \item \textit{Market-Implied Fair Price} (MIFP, denoted $\MIFP$): As in \cite{FrequentBatchAuctionsBudish}, the MIFP of a token/token swap is a publicly observable signal which is perfectly informative of the fundamental fair price of the underlying token/token swap. Moreover, a random order of fixed notional generated by a player in the system is equally likely to buy or sell tokens at the MIFP, distributed symmetrically around the MIFP. Unless otherwise stated, observing the MIFP has a prohibitive cost for players in our system. 
     
    \item \textit{Market}: A market in a DEX between two tokens $\tokenA$ and $\tokenB$ consists of two limit orders, a \textit{bid} and \textit{offer}. When the market is quoted from token $\tokenA$ to $\tokenB$, the offer price indicates the quantity of token $\tokenA$ a player must sell for 1 token $\tokenB$, while the bid price indicates the quantity of token $\tokenA$ a player receives for 1 token $\tokenB$. In this paper, we represent such a market as $\bid \ @ \ \offer$, with $0 < \bid \leq \offer $. 
     
    \item \textit{Reference Price} ($\refPrice$): For a market $\bid \ @ \ \offer$, the reference price $\refPrice$ is the price such that $\frac{\bid}{\refPrice}=\frac{\refPrice}{\offer}$, i.e., $\refPrice$ is the geometric mean of $\bid $ and $ \offer$.
       
    \item \textit{(Market) Width} ($\width$): For a market $\bid \ @ \ \offer$, the width is calculated as $\width=\frac{\offer}{\bid}$ (as such $\width\geq 1$). 
     
    \item \textit{Multiplicative Market-Impact Coefficient ($\postTradeImpact$)}: If the pre-trade MIFP for particular swap is $\MIFP$, the expected post-trade MIFP given a buy order is $\postTradeImpact \ \MIFP$ for some $\postTradeImpact\geq1$, while the expected post-trade MIFP given a sell order is $\frac{\MIFP}{\postTradeImpact}$. Unless otherwise stated, a swap from $\tokenA$ to $\tokenB$ with multiplicative market-impact coefficient $\postTradeImpact$ corresponds to buy orders of $\tokenB$ having a  multiplicative market-impact coefficient on $\MIFPB$ of $\sqrt{\postTradeImpact}$ and $\frac{1}{\sqrt{\postTradeImpact}}$ on $\MIFPA$. Given our definition of the MIFP, this impact function implies an upward drift in $\MIFP$ if $\postTradeImpact>1$. However, our use of $\postTradeImpact$ is intended to highlight that impact must be considered, with the exact choice of $\postTradeImpact$ for a particular token pair being a complex process and beyond the scope of this paper.

    \item \textit{Client}: Any player in a DEX protocol who, for an MIFP $\MIFP$, there exists some \textit{minimum client utility} $\fee >1$ such that client buyers (sellers) have positive expected utility to trade for or below $\sqrt{\fee} \ \MIFP$ (at or above $\frac{\MIFP}{\sqrt{\fee}}$).  
    
    \item  \textit{Market Maker} (MM): A player in a DEX protocol with large supplies of all tokens, who has positive expected utility trading with clients on markets of any width $\width>1$ with reference price equal to the MIFP. MMs can observe the MIFP.
    
    \item \textit{Strict Nash Equilibrium} (SNE) \cite{AlgoGameTheory}: Consider a set of non-cooperative players $\player_1$ $,...,$ $\player_n$, with strategies (series' of actions) $\strategy_1$ $,...,$ $\strategy_n$ describing the actions which each player takes throughout a particular protocol. These strategies form a strict Nash Equilibrium if any individual player deviation from these strategies strictly reduces that player's utility.

\end{itemize}

In creating a DEX protocol, an idealised goal would be to ensure that there exists an SNE in which clients trade at the MIFP in expectancy. However, this is unrealistic as MMs are a key component in liquidity provision. Therefore a more realistic, yet still desirable, goal would be to ensure that there exists an SNE where clients can trade at the MIFP in expectancy in exchange for some pre-determined fee, payable to the MMs, which is bounded by the clients' utility gain from the swap. In existing AMMs and DEX protocols, this realistic goal remains unachieved, as explained in Section \ref{sec:RelatedWork}. $\protocolName$ however, achieves this goal.

Note that MMs differ from liquidity pools in AMMs. The decision logic of AMM liquidity pools is public and deterministic, and any adjustments to liquidity pools must be queued publicly in the mempool, exposing it to EEV attacks. MMs, however, make private trading decisions and communicate them on-chain. One possible action is to add liquidity to an AMM, or in the case of $\protocolName$, add a market to an auction. Following the analysis of \cite{ImpermanentLossAMMsLoesch} and the losses being incurred by liquidity providers in AMMs, players currently providing liquidity in AMMs, although acting honestly, do not fit our rational player model. In $\protocolName$, by ensuring following the protocol forms an SNE, honesty and rationality are equivalent. If players deviate from the protocol in $\protocolName$, this strictly decreases their expected utility, which is further discussed in Appendix \ref{sec:Irrational}.

%% file: Prelims/ZKMachinery2.tex
\subsection{Zero-Knowledge Primitives}\label{sec:ZK}

The aim of this section is to outline the \textit{non-interactive zero-knowledge} (NIZK) tools for set membership as used in this paper, such as those stemming from papers like \cite{ZerocoinGreen,ZCash,PairingBasedNIZKsGroth,ZKProofsSetMembershipBenarroch,SemaphoreWhitehat}. We generalise these formal works, allowing for the adoption of any secure NIZK set-membership protocol into $\protocolName$, as we only require a common functionality that is shared by all of them. Further elaboration on these protocols is deferred to Appendix \ref{app:ZK}.

The zero-knowledge proofs used in $\protocolName$ allow, for a given set of commitments $\commitments$ to user-generated secrets, that any user knowing the secret corresponding to a commitment $\commitment \in \commitments$ can prove the knowledge of a secret corresponding to a commitment in the set, without revealing which secret, or commitment. Moreover, we require that more than one proof relating to the same commitment is identifiable by a verifier. 

To participate in $ \protocolName $, clients privately generate two bit strings, the \textit{serial number} $\commSerialNum$ and \textit{randomness} $\commRandomness$, with $\commSerialNum, \ \commRandomness \in \{0,1\}^{\Theta(\cryptoParam)}$ .
To describe $\protocolName$ we define a commitment scheme $\commit$, a set membership proof scheme \textit{SetMembership}, an NIZK proof of knowledge scheme \textit{NIZKPoK} and a NIZK signature of knowledge scheme (\textit{NIZKSoK}). We do not specify which instantiation of these schemes to use, as the exact choice will depend on
several factors, such as efficiency, resource limitations and/or the strength of the assumptions used.
 
 \begin{itemize}
     \item $\commit(m)$: A deterministic, collision-resistant function taking as input a string $m \in \{ 0,1\}^{*}$, and outputting a string $\commitment \in \{0,1\}^{\Theta(\cryptoParam)}$.
     \item \textit{SetMembership}($\commitment, \commitments$): Compresses a set of commitments $\commitments$ and generates a membership proof $\proofZK$ that $\commitment$ is in $\commitments$ if $\commitment \in \commitments$.
     \item \textit{NIZKPoK}$(\commRandomness, \commSerialNum, \commitments)$: For a set of commitments $\commitments$, returns a string $\commSerialNum$ and NIZK proof of knowledge if the person running \textit{NIZKPoK}() knows an $\commRandomness$ producing a proof when running \textit{SetMembership}($\commit(\commSerialNum || \commRandomness), \commitments$). In $\protocolName$, this revelation identifies to a verifier when a proof has previously been provided for a particular, albeit unknown, commitment\ignore{\footnote{In Z(ero)cash \cite{ZCash}, serial numbers are also used for the same purpose; to prevent double-spending of coins.} \footnote{In Semaphore \cite{SemaphoreWhitehat}, a similar technique using one-time nullifiers prevents the reuse of a commitment for protocols where one interaction per commitment is required, such as one-person one-vote voting.}} as the prover must reproduce $\commSerialNum$. This is used in $\protocolName$, in conjunction with an escrow, to enforce the correct participation of both, clients and MMs.
     \item \textit{NIZKSoK}($m$): Returns a signature of knowledge that the person who chose $m$ can also produce $\textrm{NIZKPoK}$.
 \end{itemize}

%% file: WSFBAs/WSFBAs.tex
\section{Width-Sensitive Frequent Batch Auctions}\label{sec:FBAs}

In this section we outline the properties of an idealised variation of an FBA which we define as a width-sensitive FBA. Width-sensitive FBAs maintain the desirable properties of FBAs with respect to optimal strategies for MMs and clients \cite{FrequentBatchAuctionsBudish}, while also adding important protections for clients in a decentralised setting where monopolistic MMs may exist. 
The important assumption with regard to the guarantees of an FBA is the presence of at least two non-cooperative MMs. In a decentralised setting, this can be seen as insufficient. One of the most desirable properties of FBAs in the presence of 2 non-cooperative MMs is the fact that clients submit market orders. We envisage clients as relatively uninformed players for whom choosing the correct/fair price to trade has an implicit cost. Market orders remove this burden, providing clients with an ``obvious optimal'' as advocated in \cite{EIP1559roughgarden}. To reach a similar equilibrium in the presence of a monopolistic MM, we must amend the basic FBA protocol. In this section, we define a \textit{width-sensitive} FBA ($\ourAuctionName$) to handle monopolistic MMs, while retaining the desirable properties of an FBA in the presence of two or more non-cooperative MMs.  

In the presence of a single rational MM, we need to utilise the value gained by clients for exchanging token. That is, recall from Section \ref{sec:Prelims}, clients in our protocol observe a positive utility of at least the minimum client fee $\fee$ for exchanging tokens.
In a $\ourAuctionName$, this fee is translated to a market width, and input with clients orders as a maximum market width on which clients are willing to trade. This allows us to prove submitting market orders remains a SNE. Conversely in an FBA, if MMs cooperate/are replaced by a monopolistic MM, submitting market orders is a strictly dominated strategy for clients, with clients now required to submit a limit price. $\ourAuctionName$s avoid this degradation of user experience, and the corresponding reduced probability of execution and quality of liquidity this has on FBAs. 

We let $\netTradeSize$ represent the net trade imbalance of clients in a particular instance of a $\ourAuctionName$ in terms of $\bitcoin$. A positive $\netTradeSize$ indicates a client buy imbalance (more client buyers than sellers of the swap), while a negative $\netTradeSize$ indicates a client sell imbalance. We require a finite bound on the absolute imbalance, which we denote $\maxAuctionNotional<\infty$, for the existence of optimal MM strategies. As in \cite{FrequentBatchAuctionsBudish}, we assume that $|\netTradeSize|\leq\maxAuctionNotional$, and in-keeping with the notion of an MIFP, $\netTradeSize$ is symmetric around 0 at the MIFP.
This $\maxAuctionNotional$ is used as the lower-bound on the notional of a MM's bid and offer in $\ourAuctionName$. Importantly, this ensures client orders submitted to the auction are executed (used in the proof of Theorem \ref{thm:ourFBA}).
We now define a $\ourAuctionName$.
\vspace{-0.2cm}
\begin{definition}\label{def:WSFBAuction}
    A \textit{width-sensitive frequent batch auction} ($\ourAuctionName$) involves MMs submitting markets to the TTP with total notional on the bid and offer of at least $\maxAuctionNotional$. Clients and MMs privately submit limit and market orders to the TTP including a requested maximum width from the tightest MM, above which the order is not executed. Orders are collected until a specified deadline. 
    After this deadline, client orders with requested width greater than or equal to the tightest MM width, along with a randomly-selected market from the tightest provided markets, are settled at a single clearing price which maximises the total notional traded, and then minimises the net trade imbalance.\footnote{As $\maxAuctionNotional$ is greater than the absolute client order imbalance, the clearing price must lie between the MM bid and offer} If there is more supply at the clearing price than demand, sell orders at the highest price at or below the clearing price are pro-rated based on size such that supply equals demand at the clearing price. Similarly, if there is more demand than supply at the clearing price, buy orders at the lowest price at or above the clearing price are pro-rated based on size such that demand equals supply at the clearing price. Any limit buy orders below/sell orders above the clearing price are not executed.
\end{definition} 
\vspace{-0.15cm}
The key differences between a conventional FBA and a $\ourAuctionName$ are the specification of MM widths by clients, the minimum MM notional requirement on the bid and offer, and the requirement for the clearing price to minimise the imbalance over all prices which maximise the notional traded. Minimising imbalance is a small optimisation which produces a reasonable and precise clearing price when MMs do not show width 1 markets as in an FBA. A precise algorithm for verifying a given clearing price satisfies these proprieties is included both in Algorithm \ref{alg:CPVerifier} and \cite{FairTraDEXGithubPublic}, and described in Appendix \ref{app:clearingPriceCalc}. The other amendments are intended to protect clients against monopolistic MMs, and are discussed in the proceeding section.

\input{WSFBAs/WSFBAProperties}

%% file: WSFBAs/WSFBAProperties.tex
\subsection{Properties of Width-Sensitive Frequent Batch Auctions}

In Theorem \ref{thm:ourFBA}, we identify an SNE for $\ourAuctionName$s, and show that it is equivalent to the SNE of an FBA. The case of a single monopolistic MM is more complex. First, we observe that an MM in a $\ourAuctionName$ always shows a market with reference price equal to the MIFP. In the proceeding lemmas and theorems, proofs omitted from the main body of the paper are included in Appendix \ref{app:proofs}.

\begin{lemma}\label{lem:WTTMMstrat}
    For an MM in a $\ourAuctionName$ between $\tokenA$ and $\tokenB$ with MIFP equal to $\MIFPAB=\frac{\MIFPB}{\MIFPA}$ and a client order of notional $\notional>0$, she strictly maximizes her expected utility by showing a market with reference price $\refPrice= \MIFPAB$ for any fixed width $\width\geq1$.
\end{lemma}

This result is independent of the choice of width and market-impact coefficient. However, it assumes that the MM trades with the client on either the bid or the offer. With respect to a $\ourAuctionName$ without notional restrictions and a monopolistic MM, if clients submit market orders, there are fringe cases (large imbalances) which incentivize MMs to show markets far from the MIFP. Removing these restrictions from a $\ourAuctionName$ makes for interesting future work. 

Recall clients have a strictly positive utility to exchange tokens described by the minimum client fee $\fee$, which is equivalent to being strongly incentivized to trade on a market with reference price $\refPrice$ and width $\width \leq \fee$. With this in mind, we can now apply the main result of \cite{FrequentBatchAuctionsBudish} to a $\ourAuctionName$.

\begin{theorem}\label{thm:ourFBA}
    For a $\ourAuctionName$, the strict Nash equilibria strategies given the number of non-cooperative MMs submitting markets being $ \numMMs$ are:
    \begin{itemize}
        \item $\numMMs=1$: Clients submit market orders of requested width $\fee$ and the MM shows a market of width at most $ \fee$ with reference price equal to the MIFP.
        \item $\numMMs\geq 2$: Clients submit market orders of requested width greater than 1 and MMs show a market of width 1 with reference price equal to the MIFP.
    \end{itemize}
\end{theorem}

Theorem \ref{thm:ourFBA} identifies that clients always submit market orders, and in settings where it is unclear whether there is a single monopolistic MM, or many non-cooperative MMs, it can be seen that clients always submit market orders with requested width $\fee$. 

%% file: FairTraDEX/FairTraDEX.tex
\section{\texorpdfstring{$\protocolName$}{TEXT}}\label{sec:FairTraDEX}

In Section \ref{sec:FBAs} we constructed a $\ourAuctionName$ using a TTP to enforce correct player balances, order sizes, revelation of orders, correct calculation of the clearing price and the settlement of orders. In a decentralised setting with rational players, such a TTP does not exist. However, we do have access to censorship-resistant public bulletin boards in the form of  blockchain-protocols. As discussed in the Section \ref{sec:Introduction}, these bulletin boards have many caveats such as the ordering of transactions based on transaction send time not being preserved (transaction re-ordering attacks). However, if we are able to bound the delay of updates being added to such a bulletin board (transactions being confirmed on the blockchain), we can implement a $\ourAuctionName$ in such a setting.

In this section we construct the $\protocolName$ protocol as a sequence of algorithms. We then provide a series of results regarding the incentive compatibility of these algorithms with the goal of proving $\protocolName$ instantiates a $\ourAuctionName$, and that following the protocol is an SNE. 

\input{FairTraDEX/Model}

\input{FairTraDEX/Algorithms}

\input{FairTraDEX/ProtocolProperties}

\input{FairTraDEX/Analysis}

%% file: FairTraDEX/Model.tex
\subsection{System Model}\label{sec:ThreatModel}

\begin{enumerate}
    
    \item All players $\player_1,...,\player_n$ are members of a blockchain-based distributed ledger, and a corresponding PKI.
    
    \item The ledger is represented by a linear blockchain with its state progressing by having new blocks sequentially appended. For simplicity, we assume instant finality of blocks meaning that such an appended (valid) block cannot be replaced at any later point in time.

    \item A transaction submitted by a player for addition to the blockchain (either directly or relayed) while observing blockchain height $\height$,
    is included (and thus finalised) in a block of height at most $\height+\revealTXDelay$, for some known $\revealTXDelay>0$, given that the transaction remains valid for sufficiently many intermediate ledger states.

    \item The public NIZK parameters are set-up in a trusted manner. 
    
\end{enumerate}

We do not make any assumptions regarding transaction ordering in blocks. Specifically, the order in which transactions are executed is at the discretion of the block proposer. 

If block producers are participating as MMs/clients, we need to adjust $\revealTXDelay$. Let $0<\alpha<1$ bound the fraction of blocks produced over chains of length greater than $\revealTXDelay$ by a MM responding to/the set of clients requesting trades in a particular instance of a $\auctionName$ (we need to consider all clients in a request phase, as they may all have the same direction, and as such, some positive expectancy to preventing a MM revelation). We need to increase $\revealTXDelay$ by a factor of $\frac{1}{1-\alpha}$ (similar to the methodology behind the Chain Quality property in \cite{Backbone,AsyncBackbone}).
Moreover, our property can be seen as a `block-based' variant of the time-based \textit{liveness} property defined in~\cite{Backbone,AsyncBackbone}.
An example for instant finality is Algorand~\cite{CheMic19} which stands in contrast to, e.g., Bitcoin which only guarantees eventual finality, while example of a public NIZK parameter setup is a Perpetual Powers of Tau ceremony, as used in Zcash \cite{SetUpCeremony}.

%% file: FairTraDEX/Algorithms.tex
\subsection{\texorpdfstring{$\protocolName$}{TEXT} Algorithms}\label{sec:Algos}

Each player $\playeri$ owns (has exclusive access to) a set of token balances $\balanceTokens_i$ which are stored as a global variable.
For a token $\token$, $\balanceTokens_i(\token)$ is the amount of token $\token$ that $\playeri$ owns. 
Keeping the notation from Section \ref{sec:ZK}, outputs included in round brackets () are known only to the player running the algorithm, with all other outputs posted to the public bulletin board, updating existing variables/balances where appropriate. Algorithm outputs are not signed, so players observing the output of an algorithm instance can only infer information about the player running the algorithm from public outputs and any corresponding global variable updates. 

We now outline $\protocolName$ as a set of algorithms: Setup(), Register(), CommitClient(), CommitMM(), RevealClient(), RevealMM() and Resolution(). A $\protocolName$ instance is initialised by running Setup(), and proceeds indefinitely in rounds of three distinct, consecutive phases: \textit{Commit}, \textit{Reveal} and \textit{Resolution}, each of length $\revealTXDelay$ blocks (see Section \ref{sec:ThreatModel}). For readability, we provide here the intuition to the algorithms of $\protocolName$, with a detailed explanation of each algorithm provided in Appendix \ref{app:FairTraDEXAlgos}.  

Players in the underlying blockchain protocol can enter $\protocolName$ as clients by running an instance of Register(), which for a given client deposits an escrow $\escrowClient$, and generates private information ($\commSerialNum, \ $ $ \commRandomness \in \{0,1\}^{O(\cryptoParam)}$) which is used in CommitClient() to prove that the client indeed deposited an escrow, without revealing which deposit. 

In the Commit phase, all players can run any number of CommitClient() and/or CommitMM() instances. CommitClient() generates a client order, commits to that order publicly and proves in ZK that the player deposited an escrow. If such a proof cannot be generated, or a proof has already been generated for the same $\commSerialNum$, no order can be committed. A correctly run CommitMM() instance generates a market for a prospective MM, commits to that market publicly and deposits an escrow $\escrowMM$. 

In the Reveal phase, players can run any number of RevealClient() and/or RevealMM() instances. RevealClient() publishes an order generated through CommitClient(), returning the escrow corresponding to the CommitClient() instance, and as such the Register() instance, to the client. RevealMM() publishes a market corresponding to a CommitMM() instance, and returns the corresponding escrow. Both Reveal phase algorithms assert that the client and MM have sufficient token balances to submit their order and market respectively. These assertions are also ensured in the Commit phase, but must be rechecked to ensure correct balances at the point of token transfer. 

In the Resolution phase, any number of Resolution() instances can be run. The first correct Resolution() instance selects the tightest market from the set of revealed markets,  $\revealedMkts$, for inclusion in order settlement, and any tie-breaks settled using $\commit( \revealedMkts)$, as a random seed\ignore{\footnote{Given all markets are revealed, the final value of $\revealedMkts$, and as such $\commit(\revealedMkts )$, is unpredictable in the presence of two or more non-cooperative MMs. We prove in Lemma \ref{lem:SINCEMMsGivenCommit} that all MMs running CommitMM() also run RevealMM(). The blockchain-based implementation of this function is described in Appendix \ref{app:tiebreaker}}}. The clearing price which maximises notional traded, and then minimises the notional imbalance of the remaining market and orders is computed. A precise algorithm for verifying the clearing price is included in Appendix \ref{app:ProtocolEncoding}, Algorithm \ref{alg:CPVerifier}, and described in Appendix \ref{app:clearingPriceCalc}\ignore{ The intuition behind the algorithm is as follows: Given more tokens are sold than bought at the proposed price, it can be seen that checking the next price point lower, first for higher traded notional, and then for a greater or equal absolute imbalance is sufficient to verify the proposed price is a valid clearing price. The equivalent check at the next price point above holds when more tokens are bought than sold at the proposed price.}. Orders and markets are then settled based on this clearing price. Finally, the arrays tracking active commitments, orders and markets $ \clientCommits$, $\MMCommits$, $\revealedOrders$, $\revealedMkts $ are cleared, so unsuccessfully revealed commitments during this round cannot be used to run RevealClient() or RevealMM() in future rounds. This effectively destroys the deposited escrows of such commitments.

%% file: FairTraDEX/ProtocolProperties.tex
\subsection{Properties of \texorpdfstring{$\protocolName$}{TEXT}}\label{sec:propertiesOfFairTraDEX}

We now argue that $\protocolName$ possesses all of the necessary properties to instantiate a $\ourAuctionName$, and discuss the practical implications of these properties. As the Register() and CommitClient() algorithms are constructed analogously to the Mint() and Spend() functions in \cite{ZerocoinGreen}, we can make use of the results as provided therein. These can be translated informally as:
\begin{enumerate} 
    \item Linkability: Consider a player $\player_j$, a set of registrations $\regTokens$ to which $\player_j$ does not know the privately committed values, and a valid ZK signature of knowledge $\proofZK$ and serial number $\commSerialNum$ corresponding to some $\regToken_i \in \regTokens$. $\player_j$ in has no advantage in linking $\proofZK$ and $\commSerialNum$ to the corresponding $\regToken_i$ over probability $\frac{1}{|\regTokens|}+\textit{negl}(\cryptoParam)$.
    \item Double-spending: Given a set of registrations $\regTokens$, and any number of  valid $(\proofZK, \ \commSerialNum$) pairs corresponding to elements in $\regTokens$, it is computationally infeasible to generate a new serial number $\commSerialNum'$ and corresponding valid proof of registration $\proofZK \ '$ in $\regTokens$. 
\end{enumerate} 

Given that all players in the system are registered as clients, by definition of MIFP, the expected trade imbalance implied by their orders is $0$. However, in reality, we cannot always expect this level of client participation, with less client registrations typically resulting in a greater advantage for rational players in predicting the implied trade imbalance of committed orders.

To account for this in our analysis, we introduce $n_\secParam$ denoting the minimal number of registrations required to guarantee that EEV is $\negligible(\secParam)$. Note that in certain blockchain systems, as the total number of players may be unknown to players within the system, precisely defining $n_\secParam$ may not be possible (for example, a player registering for the second time observes a smaller relative increase than a player registering for the first time). In that sense, our analysis demonstrates that the listed desirable properties can be achieved under a sufficient level of registration ($n_\secParam$ registrations), but not necessarily that a client can detect whether this level is met in a given auction instance. 
In practice, a client's decision whether or not to commit to an order in $\protocolName$ will be based on heuristics involving the number $n_c$ of observed client registrations, noting that non-negligible EEV may be tolerable if the total expected participation fees are less than $\fee$.

\ignore{ 
Given all players in the system are registered as clients, by definition of the MIFP, the expected imbalance of an order from a randomly selected client must be 0. Furthermore, as players register to the protocol, by the law of large numbers, a player $\playeri$ observing an order committed by some other non-cooperative player, the expected imbalance at the MIFP of that order approaches 0 (as commitments are indistinguishable, from the linkability property above). 

We let $ n_\secParam >1$ be such that given at least $n_\secParam$ Register() calls, from the view of any player in the protocol, the expected absolute imbalance at the MIFP of an order committed by another player is $\negligible$. Note that we have $n_\secParam$ increasing and bounded in $\secParam$. Letting $N_T $ be the total notional in the blockchain system, and $ N_M $ be the total notional identified to be controlled by MMs, it must be that $n_\secParam \leq \frac{N_T-N_M}{\escrowClient}$, $\forall \ \secParam \geq 1$. This bound is equivalent to the maximum possible number of client registrations. As such, clients seeking higher levels of EEV protection must wait for increasingly many registrations before committing to an order.  

Given a monopolistic MM, as $\psi$ increases, $n_\secParam$ quickly reaches the limit of $\frac{N_T-N_M}{\escrowClient}$. We believe that the use of $n_\secParam$ as a variable which must increase to reduce EEV and vice versa provides important intuition regarding the benefits provided by $\protocolName$. For the purpose of this paper, it is sufficient to know that increasing $n_\secParam$ reduces a MM's ability to guess which client corresponds to a committed order, and importantly, the order contents. Quantifying EEV, and it's precise relationship with $n_\secParam$ in the presence of a monopolistic MM makes for interesting future work. When two or more non-cooperative MMs exist, as all MMs can run Register() instances, inferring the true client imbalance from $\regTokens$ becomes significantly less effective, as MMs can bluff, double-bluff, etc. other MMs into an incorrect expected client imbalance.
} 

When players are rational, however, running these algorithms might not be an SNE. Only if following a protocol is an SNE can we be sure that rational players correctly follow the protocol. Towards proving $\protocolName$ forms an SNE for all rational clients and MMs, we prove a serious of Lemmas that we use to prove the main result of the section, Theorem \ref{thm:FairTraDEX}. Due to space restrictions, we provide here an intuition for these Lemmas, while formally stating and proving them in Appendix \ref{app:proofs}. We first prove that some player in the blockchain protocol runs a Resolution() instance every round. Then, we prove a serious of Lemmas demonstrating that given a rational client (resp. MM) runs an instance of Register() (resp. CommitMM()), that same player correctly runs CommitClient() and RevealClient() (resp. RevealMM()) in the proceeding phases. Finally, we show that it is indeed an SNE for a client (resp. MM) to run Register() (resp. CommitMM()). 

With these results in hand, we have it that rational clients and rational MMs correctly execute all algorithms as outlined by $\protocolName$.
We now show that with at least $n_\secParam$ Register() calls, the optimal strategy for a client is to submit market orders, while the optimal strategy for a MM with MIFP $\MIFPAB$ is to show a market $\bid \ @ \ \offer$ with $\bid \approx \MIFPAB \approx \offer$ in the case where there are at least 2 non-cooperative MMs, and of width $\width \leq \fee$ otherwise.

\begin{theorem}\label{thm:FairTraDEX}
    Consider an instance of $\protocolName$ between $\tokenA$ and $\tokenB$ with MIFP $\MIFPAB$ and at least $n_\secParam$ previously called instances of Register(). For $ \numMMs$ non-cooperative MMs, the following strategies form strict Nash equilibria:
    \begin{itemize}
        \item $\numMMs=1$: Clients run Register(), followed by CommitClient() producing market orders of width $\fee$. The MM runs CommitMM() producing a market of width at most $ \fee$ with reference price equal to $\MIFPAB$ in size $\maxAuctionNotional$. Clients and MMs then run RevealClient() and RevealMM() respectively.
        \item $\numMMs\geq 2$: 
        Clients run Register(), followed by CommitClient() producing market orders of width greater than 1. MMs run CommitMM() producing markets of width 1 with reference price equal to $\MIFPAB$ in size of at least $\maxAuctionNotional$. Clients and MMs then run RevealClient() and RevealMM() respectively.
    \end{itemize}
\end{theorem}

Although providing width-  markets may seem prohibitive for MMs, the unique guarantees of $\protocolName$ ensure that no players external to the protocol can extract value from players within the protocol. As player value is being retained within the $\protocolName$ protocol, fees can be introduced to compensate MMs. Given the potential value retention of $\protocolName$ (see Section \ref{sec:Attacks}, Table \ref{table:AttackComparison}), these fees can be substantial while still ensuring $\protocolName$ provides clients with best-in-class liquidity.

\begin{remark}\label{rem:Killer}
    To minimise expected absolute trade imbalances in a DEX auction, existing protocols, including $\protocolName$, require the hiding/mixing of order-information. Consider how $\protocolName$ compares to previous DEXs aimed at ensuring client privacy \cite{P2DEXBaum,BlockAuctionPeriodicAuctionsConstantinides,PubliclyVerifiableSecrecyPreservingPeriodicAuctionsGalal}. In these previous protocols, each order commit reveals the same, and in some cases more information per-order than a Register() call in $\protocolName$. EEV protection guarantees in these previous protocols which require $n_\secParam$ orders per auction are achieved in $\protocolName$ for every order in every auction when $n_\secParam$ players are registered to participate in the protocol. This is an $n_\secParam$ factor improvement in EEV protection/block-space requirements per auction. 
    
    More than this, these previous protocols face liveness issues when players are concerned about EEV. The first players entering one of these previous protocols must choose to do so without any guarantees of protection against EEV attacks based on information leaked from order commitment (trade direction, identity, trading patterns, etc.). 
\end{remark}

%% file: FairTraDEX/Analysis.tex
\subsection{Smart Contract Implementation of \texorpdfstring{$\protocolName$}{TEXT}}\label{sec:Encoding}

A blockchain-based pseudo-code implementation of $\protocolName$, and code description, are provided in Appendix \ref{app:Protocol}, while a Solidity implementation of Fair-TraDEX is provided in \cite{FairTraDEXGithubPublic}. We outline here the key differences between the algorithmic description of Section \ref{sec:Algos}, and the blockchain-based implementations of Appendix \ref{app:Protocol} and \cite{FairTraDEXGithubPublic}. As a blockchain-based implementation under the  model of Section \ref{sec:ThreatModel} involves a PKI for message sending, all public algorithm outputs must now be signed using the PKI. These messages must now be included in blockchain transactions, with a transaction fee required to ensure the transaction gets added to the blockchain. 

For a player to publish a transaction to a blockchain-based smart contract without revealing her identity, she must utilise a relayer (for details on relayers, see Appendix \ref{sec:Relayers}). Otherwise, the transaction fee is payable from her account, revealing sensitive information such as trading patterns and account balances. Furthermore, this relayer must be rewarded on-chain for relaying the transaction. This reward is added by the client when depositing her escrow, and retrievable by the first relayer publishing the transaction to the blockchain.
Furthermore all checks, such as those for the previous use of serial numbers in CommitClient(), or the recording of the tightest MM width in Resolution(), are explicitly encoded in the provided implementations.

%% file: Attacks.tex
\section{Cost-Benefit Analysis of \texorpdfstring{$\protocolName$}{TEXT}}\label{sec:Attacks}

The aim of this section is to demonstrate the contributory significance of FairTra-DEX vs. current state-of-the-art protocols as introduced in Section \ref{sec:RelatedWork}.
In Table \ref{table:RelatedWorkComparison} we include an overview of the gas costs for running $\protocolName$ compared to the previous blockchain-based attempts to implement batch auctions of \cite{PubliclyVerifiableSecrecyPreservingPeriodicAuctionsGalal,BlockAuctionPeriodicAuctionsConstantinides}, with numbers taken from the respective papers. It can be seen that $\protocolName$ has a slightly greater upfront gas cost for clients, but a lesser cost for MMs. 

To demonstrate the benefits of $\protocolName$, Table \ref{table:AttackComparison} compares specific swaps that allow for EEV attacks in existing state-of-the-art protocols. 
We perform our analysis on ETH/USDC swaps, as this is the highest volume pool on Uniswap, which at time of writing had pool sizes of 120k ETH and 185M USDC, an indicative MIFP of 1 ETH equal to 1,540 USDC \cite{UniswapWebsite}. Furthermore, we use a gas cost of 7 gwei \cite{ETHGasCost}. 

Consider 3 buy ETH orders of 10k, 500k and 10M USDC from 2 different players who are known to need to trade at any price. $\player_1$ has large quantities of both ETH and USDC, and buys or sells ETH pseudo-randomly, while $\player_2$ only owns USDC/only buys ETH. We take the estimated impact for each order to be $~0$, $0.15\%$  and $1\%$ respectively, numbers taken from the Uniswap V3 API \cite{UniswapWebsite} (these are more realistic impacts than those implied by the constant product impact \cite{AnalysisOfUniswapAngeris} of $~0$, $0.54\%$  and $11.1\%$ respectively). Although this is a simplification of order impact, true impact is likely some multiple/factor of this impact. Protocol fees incentivizing MMs to provide liquidity are omitted as they are not considered in the provided academic protocols. After gas costs, this fee should be approximately equal for all protocols (the Uniswap fee for this pool is $0.3\%$).

\vspace{-0.4cm}

\begin{table*}[h]
\centering
\begin{tabular}{l|c|c|c|c}
    \toprule
     & $\protocolName$ \footnotemark  & Uniswap & \cite{BlockAuctionPeriodicAuctionsConstantinides} & \cite{PubliclyVerifiableSecrecyPreservingPeriodicAuctionsGalal} \\
    \midrule
    Register & 112,800 & -  & 87,000 & -  \\  
    Commit Client & 344,500 & -  & 52,000 & 276,150  \\ 
    Commit MM (per order) & 24,300 & -   & 52,000 & 276,150 \\  
    Reveal (per order) & 172,000 & 190,000  & 171,000 & 48,750 \\  
    Settle (per order) \footnotemark   & 45,500 & -  & 122,500 & 54,000  \\ 
    \midrule
    Total Client  & 674,800 & 190,000 & 432,500 & 378,900 \\
    Total MM & 266,100  & -    & 397,500 & 649,050 \\
    \midrule
    Total Client (USDC) & 7.27 & 2.05 & 4.66 & 4.08 \\
    Total MM (USDC) & 2.87 & - & 4.09 & 7 \\
    \bottomrule
  \end{tabular}
  \caption{Comparison of gas costs in batch-auction implementations\label{table:RelatedWorkComparison}.
  \footnotesize{\textsuperscript{2} Costs provided for $\protocolName$ are amortised over 128 client orders and 8 markets. \textsuperscript{3} We add an estimated cost for token transfer from smart contract to player of 40,000 to the figures provided in \cite{PubliclyVerifiableSecrecyPreservingPeriodicAuctionsGalal} to standardise the costs therein with those of $\protocolName$ and \cite{BlockAuctionPeriodicAuctionsConstantinides}.} }
\end{table*}

\vspace{-1.8cm}

\begin{table*}[h]
\centering
\begin{tabular}{l|c|c|c|c}
    \toprule
     & $\protocolName$  & Uniswap & \cite{BlockAuctionPeriodicAuctionsConstantinides} & \cite{PubliclyVerifiableSecrecyPreservingPeriodicAuctionsGalal}  \\
    \midrule
    $\player_1$-10,000 & 0 & 50 & 0 & 0  \\
    $\player_2$-10,000 & 0 & 50 & 0 & 0 \\
    $\player_1$-500,000 & 0 & 3000 & 750 & 0 \\
    $\player_2$-500,000 & 0 & 3000 & 750 & 750 \\
    $\player_1$-10,000,000 & 0 & 150,000 & 100,000 & 0  \\
    $\player_2$-10,000,000 & 0 & 150,000 & 100,000 & 100,000 \\ 
    \bottomrule
  \end{tabular}
  \caption{Comparison of execution costs in USDC of batch-auction implementations. \label{table:AttackComparison}}
\end{table*}

\vspace{-1cm}

When $\player_1$ submits an order in $\protocolName$ or \cite{PubliclyVerifiableSecrecyPreservingPeriodicAuctionsGalal}, no information is gained about the direction of the trade. However, in \cite{BlockAuctionPeriodicAuctionsConstantinides}, direction is revealed. As such, any blockchain participant can front run that impact on all other markets, and thus the MIFP for any MM responding to the order will be the impacted MIFP. When $\player_2$ submits an order in either of \cite{BlockAuctionPeriodicAuctionsConstantinides,PubliclyVerifiableSecrecyPreservingPeriodicAuctionsGalal} the direction is known, and the MIFP is impacted in the same way as for $\player_1$, \textbf{before} any player interacts with $\player_2$, giving $\player_2$ a worse price. Using estimated price impacts of $~0$, $0.15\%$  and $1\%$, Table \ref{table:AttackComparison} demonstrates the costs of executing these swaps, excluding transaction fees, in these protocols, and Uniswap. For Uniswap, we must also add the recommended slippage, an additional $0.5\%$ of the order size, as it is always in a block producers interest to give Uniswap players worst execution. It can be seen that these costs become increasingly more significant as order size increases, dominating the differences in gas costs of Table \ref{table:RelatedWorkComparison}.

Although Table \ref{table:AttackComparison} can be seen as simplifying how orders are handled, it demonstrates two crucial motivators for our work. Firstly, any information revealed about clients before a trade is agreed can, is and will continue to be used against clients. Furthermore, this cost is not necessarily paid to the MM. As orders are committed in public, any blockchain participant can use the committed information to front run the impact on the MIFP before the client or MM has an opportunity to trade, extracting money from the DEX protocol. Secondly, as the effects of these value-extraction techniques increase super-linearly in order-size, a protocol with the value-extraction guarantees of $\protocolName$ is needed to allow typically large clients to utilise the benefits of DEXs, and blockchain protocols in as a whole, at a fixed cost, as demonstrated in Table \ref{table:RelatedWorkComparison}, without incurring the prohibitive execution costs of previous solutions, as demonstrated in Table \ref{table:AttackComparison}.

%% file: Conclusion.tex
\vspace{-0.5cm}

\section{Conclusion}\label{sec:conclusion}

We provide $\protocolName$, a blockchain-based DEX protocol based on WSFBAs in which we formally prove the strategies of rational participants have strict Nash equilibria in which all trades occur at the MIFP plus or minus bounded upfront costs (specified market widths) which approach 0 in the presence of non-cooperative MMs. This is an attractive alternative to existing mainstream protocols such as AMMs where rational players effectively and systematically prevent such an equilibrium from happening. Compared to previous blockchain-based attempts to implement EEV-proof DEXs, $\protocolName$ is the first to practically allow for indistinguishable client-order submissions by decoupling order submission from escrow deposit and order revelation. The $\protocolName$ benefits formalised in Section \ref{sec:propertiesOfFairTraDEX}, summarised in Remark \ref{rem:Killer}, and demonstrated in Section \ref{sec:Attacks} provide important improvements on previous protocols regarding EEV protection, setting a new standard for EEV protection in DEXs.

As stated in the comparisons of Section \ref{sec:Attacks}, protocol fees are omitted for all protocols. Given the total retention of value within the $\protocolName$ protocol (no extractable value), fees in line with the utility gained by clients for exchanging their tokens can be charged to incentivise the long-term participation of MMs in $\protocolName$. These fees should reflect the need to incentivize MMs while retaining the unique client-side benefit of trading at the MIFP in expectation, which is proven to occur in $\protocolName$. Analysis of these fees makes for interesting future work.

\ignore{To do this, we show that $\protocolName$ implements a width-sensitive FBA, an FBA variant that benefits from the game-theoretic guarantees of an FBA in a distributed blockchain-based setting. We prove our construction is strong incentive compatible in expectation for all rational participants, and combine these results in Theorem \ref{thm:FairTraDEX} which outlines the conditions under which clients and MMs trade at the MIFP. In this equilibrium, all fees are known a-priori, with EEV effectively prevented when enough players have registered to participate in the protocol. 
Analysing and optimising the provided $\protocolName$ implementations for wide-scale deployment makes for exciting future work. We also envisage a paradigm in which existing work on delay encryption \cite{DelayEncryptionBurdges} and/or threshold-key encryption \cite{ThresholdEncryptionAdaptiveDomingoFerrer} becomes sufficiently developed in fully decentralized systems, particularly in the ByRa model \cite{AchievingSMRwithoutHonestPlayers}. Advances in either of these areas could allow $\protocolName$ to become a single-action protocol for clients and MMs, with an equivalent user experience to current AMMs in addition to provable EEV prevention.}

%% file: Appendices/Appendices.tex
\input{Appendices/RelWProtocols}

\input{Appendices/Terminology}

\input{Appendices/ZKAddOns}

\input{Appendices/Proofs}

\input{Appendices/FormalAlgos}

\input{Appendices/ClearingPriceCalc}

\input{Appendices/ContractDescription}

\input{Appendices/Consideration}

\input{Appendices/NotesOnFTDEX}

%% file: Appendices/RelWProtocols.tex
\section{Extended Related Work}\label{app:RelatedWork}

\subsection*{Estimating (Miner) Extractable Value is Hard \cite{EstimatingMEVLetsGoShoppingJudmayer}}

This paper paper attempts to formalise extractable value and generalise it beyond value extractable be miners. We also believe it is necessary to model the decision of all rational players based on \textit{expected extractable value} (EEV) that can generated by particular orderings of transactions/blocks by any player in the system, and not just the miner. The approach taken is to consider EEV as the maximum of all non-protocol strategies, with protocols considered secure if the EEV of following the protocol is strictly dominated by following the protocol strategy, which is further formalised in \cite{AchievingSMRwithoutHonestPlayers}. In our paper, we also consider an additional case of EEV not necessarily considered in \cite{EstimatingMEVLetsGoShoppingJudmayer} which is prevalent in commit-reveal protocols such as \cite{FairMMCiampi}. In such protocols, honest behaviour usually involves sending a valid second transaction (the reveal transaction in a commit-reveal protocol), but where players can extract value in expectancy by not sending these transactions. However, we believe the definition of EEV in \cite{EstimatingMEVLetsGoShoppingJudmayer} can be extended to include these attacks. As such, we also move away from the legacy use of MEV, and focus instead on the prevention of the more general EEV.

\subsection*{Publicly Veriﬁable and Secrecy Preserving Periodic Auctions \cite{PubliclyVerifiableSecrecyPreservingPeriodicAuctionsGalal}}

This protocol also attempts to implement an FBA, and as such has many similarites to $\protocolName$. As in $\protocolName$, the protocol progresses in rounds of Commit, Reveal and Resolution phases. In \cite{PubliclyVerifiableSecrecyPreservingPeriodicAuctionsGalal}, there is a designated \textit{operator} who is in charge of settling the auction. Players commit to orders in the Commit phase, as well as providing cryptographic information, which is used to prove correct settlement in the Resolution phase. Unlike $\protocolName$, these commit messages are sent by players directly to the blockchain, revealing identity and trade direction. In the Reveal phase, players encrypt their orders using the operator's public key, and send the encryptions to the operator. In the Resolution phase, the operator then chooses a clearing price which intersects the buy and sell liquidity, maximising the notional to be traded. The operator then publishes a list of all matched orders to the blockchain, along with a range proof which is used to verify the correct execution of orders, while not revealing any information about unexecuted orders other than that already revealed in the commit phase.  

\subsection*{Block Auction \cite{BlockAuctionPeriodicAuctionsConstantinides}}

This protocol attempts to implement an FBA, and is the most similar to $\protocolName$. It is an improvement on \cite{PubliclyVerifiableSecrecyPreservingPeriodicAuctionsGalal}, with a direct comparison of the two protocols forming the main basis of the justification of \cite{BlockAuctionPeriodicAuctionsConstantinides}. As in \cite{PubliclyVerifiableSecrecyPreservingPeriodicAuctionsGalal}, the protocol is overseen by an \textit{operator} who is in charge of receiving orders privately from players and correctly executing the auction. As in $\protocolName$, the protocol progresses in rounds of Commit, Reveal and Resolution phases. In the Commit phase, players commit to orders and publish these commitments to the blockchain. Although not revealing the trade direction as is the case in \cite{PubliclyVerifiableSecrecyPreservingPeriodicAuctionsGalal}, these commit messages are sent by players directly to the blockchain, and as such reveal identity  In the Reveal phase, players encrypt their orders using the operator's public key, and send these encryptions to the operator.  The operator then publishes all executed orders, while revealing nothing about unexecuted orders. The validity of execution depends on all players who submitted orders verifying that their order should not have been executed given the list of executions.

\subsection*{P2DEX \cite{P2DEXBaum}}

The P2DEX protocol is an off-chain MPC protocol run by servers where players can submit orders to exchange tokens from one blockchain to another (although it also appears applicable to one blockchain with many tokens). The orders are encrypted using a threshold secret-sharing scheme with each server receiving a unique share. The protocol has mechanisms to identify double-spending of player funds sent to the servers, and deviation (failure/ misbehaviour) of servers, as the MPC matching protocol is publicly verifiable. As such, all players in the blockchain can verify that a set of orders have been matched correctly, or some of the servers deviated from the protocol. The exchange depends on all servers participating in a secret-sharing protocol to match orders, with at least one server being honest/not colluding with other servers. As with \cite{PubliclyVerifiableSecrecyPreservingPeriodicAuctionsGalal,BlockAuctionPeriodicAuctionsConstantinides}, an emphasis is placed on not revealing unmatched orders. 

\subsection*{FairMM \cite{FairMMCiampi}}

Clients submit orders to a single (monopolistic) MM in an off-chain $\Sigma$-protocol. Clients contact the MM with a size, direction (communicated on-chain) and price (communicated off-chain). If the MM accepts, the MM then publishes the trade on the blockchain, otherwise aborts. Client orders are sequentialized so only one order can be executed at a time, preventing the MM from reordering client orders. 

\subsection*{FuturesMEX \cite{FuturesMEXMassacci}}

The FuturesMEX protocol is an MPC version of a DEX with claims of anonymity. In \cite{FuturesMEXMassacci}, client token balances are kept privately by owners in an off-chain database, so the protocol has limited applicability to a blockchain-based setting. Furthermore, orders are submitted publicly to all participants before being settled. For smaller clients with less connectivity, this is equivalent to showing the order to more-connected counterparties before it is executed. This is a typical source of EEV in existing blockchain protocols (through front-running attacks), and something which is protected against in $\protocolName$.

%% file: Appendices/Terminology.tex
\section{Terminology and Useful Definitions}\label{app:terminology}

This section contains additional financial and game-theoretical terms used in this paper. Although not mandatory for all readers, this section serves as a useful reference point towards understanding the results and discussions that follow.

\begin{itemize}
    \item \textit{Decentralized Exchange} (DEX): A distributed marketplace which allows players to swap one token for another. 
     
     \item \textit{Limit Order}: Specifies an amount of tokens to be bought (sold), and a maximum (minimum) price at which to buy (sell) these tokens. This price is known as the \textit{limit price}.
     
     \item \textit{Market Order}: Specifies an amount of tokens to be sold, but no limit price. Market orders are to be executed immediately at the best available price based on the liquidity of buy orders. 
     
     \item \textit{Direction}: With respect to an order on a market quoted from token $\tokenA$ to $\tokenB$, if the order is trying to buy token $\tokenB$, the direction is \textit{buying}, while if the order is trying to sell token $\tokenB$, the direction is \textit{selling}. 

    \item \textit{Forward Price}: This is the price at which a seller delivers a token to the buyer at some predetermined date. In any exchange protocol without instantaneous delivery, the forward price at expected delivery time is the price at which trades should happen. The difference between current (spot) price and forward price is known as \textit{carry}, and can be due to storage/opportunity costs, interest rates, etc. In this paper, we set carry to 0 for complexity and ease-of-notation purposes.

\end{itemize}

The following definition of expected extractable value is translated from \cite{EstimatingMEVLetsGoShoppingJudmayer} using the terminology of this paper.

\begin{definition}\label{def:EEV}
    The expected extractable value $\text{EEV}_i$, describes the total value in value units, which is transferred to player $\player_i$ in expectation using a certain strategy which produces a transaction, sequence of transactions, or blocks that later become part of the main chain with some probability.
\end{definition}

\subsection{Frequent Batch Auctions}

As stated in Section \ref{sec:Introduction}, $\protocolName$ is based on an FBA \cite{FrequentBatchAuctionsBudish}. FBAs are used in many of the largest centralised exchanges \footnote{FCA \url{https://www.fca.org.uk/publications/research/periodic-auctions}, CBOE \url{https://www.cboe.com/europe/equities/trading/periodic_auctions_book/}, ESMA \url{https://www.esma.europa.eu/sites/default/files/library/esma70-156-1035_final_report_call_for_evidence_periodic_auctions.pdf}}.  As FBAs were initially intended for a centralised setting, we consider them being run by a trusted third party (TTP) who enforces the correct participation of all players. In $\protocolName$, the key TTP functionalities needed to instantiate an FBA are replicated using ZK set-membership proofs, incentivisation and a blockchain as a censorship-resistant bulletin-board. We define an FBA here using the terminology of our paper. 

\begin{definition}\label{def:FBAAuction}
    A \textit{frequent batch auction} ($\auctionName$) (sometimes referred to as a periodic auction) involves clients and MMs privately submitting either limit or market orders to the TTP. These orders are collected until a specified deadline. After this deadline, the orders are settled at the clearing price. A single \textit{clearing price} is chosen which maximises the total notional traded based on the specified sizes and prices of all orders. If there is more supply (quantity of tokens being sold) at the clearing price than demand (quantity of tokens being bought), all sell orders offered at the highest price at or below the clearing price are pro-rated based on size such that supply equals demand at the clearing price. Similarly, if there is more demand than supply at the clearing price, all buy orders bid at the lowest price at or above the clearing price are pro-rated based on size such that demand equals supply at the clearing price. Any limit buy orders below/sell orders above the clearing price are not executed.
\end{definition}

There are two key differences between this definition and the specification in \cite{FrequentBatchAuctionsBudish}:
\begin{enumerate}
    \item In our definition, if an order is not fulfilled, it is revealed with any tokens not being sold returned to the seller. This does not affect the game-theoretic guarantees of the paper, as the results in \cite{FrequentBatchAuctionsBudish} only depend on the hiding of order information while players are submitting orders to the auction. 
    \item As every order in our auction must be submitted independently for each auction, there is no time priority applied when pro-rating orders in case of a supply-demand imbalance. This is a sub-case of the FBAs as defined in \cite{FrequentBatchAuctionsBudish}, and consequently, our protocol retains the same game-theoretic guarantees.
\end{enumerate}

\input{Appendices/GTGuaranteesOfFBAs}

\input{Appendices/Relayers}

%% file: Appendices/GTGuaranteesOfFBAs.tex
\subsubsection{Game-Theoretic Guarantees of a Frequent Batch Auction} \label{sec:FBAGTGuarantees}

In this section we investigate the properties of an $\auctionName$ between rational MMs and rational clients, where MMs do not know the desired trade direction of the clients. 

We first restate, using the terminology from this paper, the main result from \cite{FrequentBatchAuctionsBudish} which applies to our game-theoretically equivalent definition of an $\auctionName$. To do this, we let $\netTradeSize$ represent the net trade imbalance of clients in a particular instance of an $\auctionName$ in terms of $\bitcoin$. A positive $\netTradeSize$ indicates a client buy imbalance (more client buyers than sellers of the swap), while a negative $\netTradeSize$ indicates a client sell imbalance. We require a finite bound on the absolute imbalance, which we denote $\maxAuctionNotional<\infty$, for the existence of optimal MM strategies. As in \cite{FrequentBatchAuctionsBudish}, we assume that $|\netTradeSize|\leq\maxAuctionNotional$, and in-keeping with the notion of an MIFP, $\netTradeSize$ is symmetric around 0 at the MIFP.

\begin{theorem}\label{thm:FBA}
    \cite{FrequentBatchAuctionsBudish} For an $\auctionName$ with at least two non-cooperative MMs, there is a strict Nash equilibrium where clients only submit market orders and MMs show a market of width 1 ($\bid=\offer$) centred around the MIFP in size greater than $\maxAuctionNotional$. 
\end{theorem}

This is a useful result in the case of at least two non-cooperative MMs, with clients receiving a game-theoretic guarantee that they can exchange one token for another at the MIFP in expectancy in an $\auctionName$. Furthermore, as MM liquidity is greater than the net client trade size, the implicit impact to these trades in \cite{FrequentBatchAuctionsBudish} is bounded by the width, which is 1. As clients have a strictly positive utility for exchanging tokens, this is equivalent to clients always having positive expectancy to participate in an $\auctionName$. However, it is also shown in this equilibrium that MMs have 0 expected utility. A basic adjustment to the protocol in that setting would then be to charge clients a fee for the service and pro-rate these fees to the MMs to ensure the long-term participation of MMs.

%% file: Appendices/Relayers.tex
\subsection{Relayers}\label{sec:Relayers}

A fundamental requirement for transaction submission in blockchains is the payment of some transaction fee to simultaneously incentivise block producers to include the transaction, and to prevent denial-of-service/spamming attacks. However, in both the UTXO- and account-based models, this allows for the linking of player transactions, balances, and their associated transaction patterns. With respect to DEX protocols, if clients are required to deposit money into a UTXO/account before initiating a trade, any other player in the system can infer who the client is, what balances the client owns, what transactions the client usually performs, etc., and use this information to give the client a worse price. 

To counteract this, we utilise the concept of \textit{transaction relayers}\footnote{Ox \url{https://0x.org/docs/guides/v3-specification}, Open Gas Station Network \url{https://docs.opengsn.org/}, Rockside \url{https://rockside.io/}, Biconomy \url{https://www.biconomy.io/}}. In the smart-contract encoding of $\protocolName$ (App. \ref{app:ProtocolEncoding}), clients must publicly register to a smart contract, and in doing so, deposit some escrow. In addition to this escrow, we also require the clients to deposit a relayer fee. When the client wishes to submit a transaction anonymously to the blockchain, the client publishes a proof of membership in the set of registered clients to the relayer mempool, as well as the desired transaction and a signature of knowledge cryptographically binding the membership proof to the transaction, preventing tampering. As the relayer can verify the proof of membership, the relayer can also be sure that if the transaction is sent to the $\protocolName$ contract, the relayer will receive the corresponding fee. With this in mind, a relayer observing the client transaction includes it in a normal blockchain transaction, with the first relayer to include the transaction receiving the fee. As such, relayers are a straightforward extension of the standard transaction-submission model. Furthermore, if the proof of membership is NIZK and the message is broadcast anonymously (using the onion routing (Tor) protocol\footnote{\url{https://www.torproject.org/}} for example), the relayer can only infer that the player sending the transaction is a member of the set of clients.

%% file: Appendices/ZKAddOns.tex
\section{Background on Zero-Knowledge Primitives}\label{app:ZK}

Proving membership has been traditionally solved using cryptographic accumulators \cite{BenalohDeMare93}, where the prover $P$ computes a value (the accumulator) and, based on it, a set of short membership proofs that the verifier $V$ can easily verify. Three are the approaches to construct set membership proofs: Merkle trees \cite{Merkle87}, RSA accumulators \cite{BaricPfitzman97,BBF19}, and pairing-based accumulators \cite{Ngu05,Zhangetal17}.

Each approach has its own benefits for public parameters, accumulator or witness size or need of trusted setup. The exact choice depends on the resource constraints of the system. We direct interested readers to \cite{ZKProofsSetMembershipBenarroch} for a nice review of the main features of each of the approaches. 

When the prover $P$ does not want to reveal the value of $x$, the membership proof should not leak any information on the value of $x$. At a high level, the general approach is to guarantee privacy using zero knowledge proofs. Zero knowledge proofs \cite{seminalZKP85} are powerful cryptographic primitives that allow  a prover $P$  to prove knowledge of the truth of some statement  without revealing the statement contents, to some honest  verifier $V$ who needs to be convinced of the truth of the statement provided by the prover. 
Special mention for its applicability should be made of zero-knowledge proofs that are also non-interactive, that is,  proofs that only depend on the prover's private information about the statement and publicly available information\footnote{This public information can come in many forms, but in \cite{ZerocoinGreen,ZCash,PairingBasedNIZKsGroth,ZKProofsSetMembershipBenarroch}, it must be generated honestly in a process known as a \textit{trusted setup}. If a prover knows the private information used to generate public proof parameters, the knowledge extraction property cannot exist.}. As such, proofs do not depend on interaction with the verifier.  The main features in a NIZK argument are completeness,  soundness, and zero-knowledge. Completeness guarantees that if the statement is true, the prover behaving honestly can convince the verifier that the statement is true, while soundness ensures that a dishonest prover cannot convince an honest verifier. Zero-knowledge maintains that the only information learned by the verifier is that the statement is true. However, in practice it is required that the prover knows a witness for the statement, that is, a zero-knowledge proof of knowledge. In this case, soundness is not enough and is required that a prover cannot produce a valid proof unless she knows a witness, even if the prover has seen an arbitrary number of simulated proofs. This is what is known, as simulation intractability. 
 Furthermore, NIZK arguments are interesting for constructing other cryptographic primitives, such as signature of knowledge (SoK) \cite{SeminalSoK06}. 
 
In literature, there are several constructions  that add the privacy layer using zero knowledge proofs for set membership based on RSA Accumulators or Merkle Trees are \cite{CamenischLysyanskaya02,ZerocoinGreen,ZCash}.  In these works the prover proves statements about values that are committed, that is, they follow what is known as a commit-and-prove zero knowledge proof. More recent approaches  propose new commit-and-prove for set-membership based on SNARKs \cite{Geppeto15,ZKProofsSetMembershipBenarroch} or Bulletproofs \cite{PriBank22}. 

%% file: Appendices/Proofs.tex
\section{Proofs}\label{app:proofs}

\subsubsection*{Lemma 1.}\textit{ For an MM in a $\ourAuctionName$ between $\tokenA$ and $\tokenB$ with MIFP equal to $\MIFPAB=\frac{\MIFPB}{\MIFPA}$ and a client order of notional $\notional>0$, she strictly maximizes her expected utility by showing a market with reference price $\refPrice= \MIFPAB$ for any fixed width $\width\geq1$.}

\begin{proof}
    Let us define the market in terms of $\refPrice$ and $\width$ as described in Section \ref{sec:Prelims}, namely, $\frac{\refPrice}{\sqrt{\width}} \ @ \ \sqrt{\width} \refPrice$. In the cases of a client buyer and client seller of the swap, we convert client notional orders into the tokens being sold, mark the trades to their respective MIFPs using a multiplicative market-impact coefficient for the token swap of $\postTradeImpact$, then reconvert the tokens into notional. 
    
    If the client is a buyer of the swap, the client is selling $\tokenA$, with trade size of $\frac{\notional}{\MIFPA}$ in $\tokenA$. The trade occurs on the token swap offer of $\sqrt{\width}\refPrice$, resulting in the sale of $\frac{\notional}{\MIFPA}\frac{1}{\sqrt{\width}\refPrice}$ token $\tokenB$s by the MM. Finally, the $\tokenB$s bought by the client have an expected per-token value of $\sqrt{\postTradeImpact}\MIFPB$, while the notional acquired by the MM ($\tokenA$s) has an expected value of $\frac{\notional}{\sqrt{\postTradeImpact}}$. This is an expected net profit for the MM measured in notional of:
    \begin{equation}\label{eq:buyerProfit}
        \frac{\notional}{\sqrt{\postTradeImpact}}-\frac{\notional}{\MIFPA \sqrt{\width}\refPrice}\sqrt{\postTradeImpact}\MIFPB =  \frac{\notional}{\sqrt{\postTradeImpact}}-\frac{\sqrt{\postTradeImpact}\notional}{\sqrt{\width}} \frac{\MIFPAB}{\refPrice}.
    \end{equation}

    If the client is a seller of the swap, the client is selling $\tokenB$, with trade size of $\frac{\notional}{\MIFPB}$ in $\tokenB$. The trade occurs on the token swap bid of $\frac{\refPrice}{\sqrt{\width}}$, resulting in the sale of $\frac{\notional}{\MIFPB} \frac{\refPrice}{\sqrt{\width}}$ token $\tokenA$s by the MM. Finally, the $\tokenA$s bought by the client have an expected per-token value of $\sqrt{\postTradeImpact}\MIFPA$, while again, the notional acquired by the MM ($\tokenB$s) has an expected value of $\frac{\notional}{\sqrt{\postTradeImpact}}$  This is an expected net profit for the MM of:
    \begin{equation}\label{eq:sellerProfit}
        \frac{\notional}{\sqrt{\postTradeImpact}}-\frac{\sqrt{\postTradeImpact}\notional}{\sqrt{\width}}\frac{\refPrice}{\MIFPAB}.
    \end{equation}

    We know the expected buying and selling of $\notional$ notional at $\MIFPAB$ are equally likely by the definition of an MIFP as a perfectly-informed signal. Therefore, the total expected profit is:
    \begin{equation}\label{eq:fee}
        \notional\Big(\frac{1}{2}(\frac{1}{\sqrt{\postTradeImpact}}-\frac{\sqrt{\postTradeImpact}}{\sqrt{\width}} \frac{\MIFPAB}{\refPrice}) +\frac{1}{2}(\frac{1}{\sqrt{\postTradeImpact}}-\frac{\sqrt{\postTradeImpact}}{\sqrt{\width}}\frac{\refPrice}{\MIFPAB})\Big).
    \end{equation}
    
    To find the maximum with respect to $\refPrice$, we take the first derivative of this formula, and let it equal to 0:
    \begin{equation}
        \frac{\MIFPAB}{\refPrice^2} -\frac{1}{ \MIFPAB}=0.
    \end{equation}
    Solving for $\refPrice$ gives $\refPrice=\MIFPAB$, which is equivalent to the MM strictly maximizing her expected profits by letting $\refPrice=\MIFPAB$.
\end{proof}

\begingroup
\def\thetheorem{\ref{thm:ourFBA}}
\begin{theorem}
    For a $\ourAuctionName$, the strict Nash equilibria strategies given the number of non-cooperative MMs submitting markets being $ \numMMs$ are:
    \begin{itemize}
        \item $\numMMs=1$: Clients submit market orders of requested width $\fee$ and the MM shows a market of width at most $ \fee$ with reference price equal to the MIFP.
        \item $\numMMs\geq 2$: Clients submit market orders of requested width greater than 1 and MMs show a market of width 1 with reference price equal to the MIFP.
    \end{itemize}
\end{theorem}
\endgroup

\begin{proof}
    We now investigate each of the cases described in the theorem statement in terms of the number of non-cooperative MMs $\numMMs$.

    \textbf{$\numMMs =1$}: Consider first the strategy of a client. Although clients are not necessarily aware of the MIFP, let us consider their strategies taking the MIFP $\tokenPrice$ as a variable with an arbitrary distribution. For buy orders, the strategy of submitting a limit order with price $\tokenPrice$ less than $\sqrt{\fee}\MIFP$ is dominated by all prices greater than $\tokenPrice$ and less than or equal to $\sqrt{\fee}\MIFP$. For limit sell orders, this limit is $\frac{\MIFP}{\sqrt{\fee}}$. As such, the equilibrium for clients involves submitting orders equivalent to a market of width equivalent to at least $\fee$ with reference price equal to the MIFP. If a client knows a MM submits a market of width less than or equal to $\fee$ with reference price equal to the MIFP, this strategy is further dominated by submitting a market order with requested width $\fee$, as market orders strictly increase the client's probability of trading. Furthermore, any strategy for a client which involves trading on a price outside $[\frac{\MIFP}{\sqrt{\fee}},\sqrt{\fee}\MIFP]$ is strictly dominated by not trading. As such, the only possibilities for equilibria can occur on a market of $\frac{\MIFP}{\sqrt{\fee}} @ \sqrt{\fee}\MIFP$. Furthermore, the submission of market orders (increasing probability of trading) with requested width $\fee$ is strictly dominant if the MM shows a market with reference price equal to the MIFP in sufficient size to fill all clients' orders. As $\maxAuctionNotional>|\netTradeSize|$, this would be the case given the appropriate reference price. In Lemma \ref{lem:WTTMMstrat}, we have seen that a MM trading on a market against a random client shows a market with reference price equal to the MIFP. Therefore clients submit market orders with requested width $\fee$. Moreover, this strategy does not require the client to know the MIFP..

    Consider now the MM strategy. As only the tightest market in every auction is included for settlement, the MM only submits one market. Any MM order bidding above/offering below the MIFP has negative expected utility, and as such, no rational MM does this this. Also, by definition, the MM must show a market in size $\maxAuctionNotional\geq|\netTradeSize|$, meaning the MM has sufficient notional on the bid and offer to trade all client orders and as such the clearing price must be inside the provided MM market.
    Next, we have seen in Lemma \ref{lem:WTTMMstrat} that a MM trading on a market against a random client shows a market with reference price equal to the MIFP. Furthermore, from Equation \ref{eq:fee} we can see that the expected utility of a MM is strictly increasing in width. Any strategy involving a market with width greater than $\fee$ cannot be an equilibrium as clients strictly prefer to trade on markets of lesser width, as argued above. Therefore, the MM maximises her expectancy against a random client by showing a market of width $\fee$ with reference price equal to the MIFP. Against multiple clients, a positive notional imbalance at the MIFP is decreasing in price (resp. a negative notional imbalance is increasing as price decreases), which may cause the MM to provide a market of width less then $\fee$. 
    
    Consider the strategy of a MM providing a market of width less than or equal to $\fee$ with reference price equal to the MIFP, and the strategy of clients submitting market orders with requested width $\fee$. We have shown that any player deviation strictly decreases that player's expectancy, making this a strict Nash Equilibrium.

   \textbf{ $\numMMs\geq2$:} 
   As MMs in a standard FBA provide markets of width 1 when the width is not a restriction, applying the requested width adjustments of a $\ourAuctionName$, further incentivising tighter markets, does not change the unique equilibrium of Theorem \ref{thm:FBA}. Similarly, as there is a unique clearing price when a width-1 market is submitted, it must also minimise the imbalance over prices that maximise total notional traded. 
   The restriction on the notional of markets in a $\ourAuctionName$ is in line with the inequality of Theorem \ref{thm:FBA}. Furthermore, any requested width $>1$ in this equilibrium ensures a client's order trading through the MIFP is included in the final auction settlement, with the maximum allocation occurring when a client submits a market order. 
\end{proof}

\begin{lemma}\label{lem:ResolutionGetsRun}
    At least one player runs Resolution() in every round.
\end{lemma}

\begin{proof}
    Consider a Resolution phase where at least one player has not called a CommitClient() or CommitMM() instance in the preceding Commit phase. This player is indifferent to the settlement of orders, and as such the only payoff for that player by running Resolution() correctly is the receipt of $\resBounty \in \mathbb{R}^+$. 
    
    Consider instead the case where all players in the system called at least one instance of CommitClient() and/or CommitMM() in the preceding Commit phase. In this case, all players have an additional payoff for receipt of the tokens currently locked in the protocol. As at least one of the buyers or sellers of the swap must receive a non-zero amount of tokens, the receipt of which having at worst 0 utility. This, in addition to the receipt of $\resBounty$ makes the calling of Resolution() positive expectancy for at least one player in the system. 
\end{proof}

\begin{lemma}\label{lem:SINCEClientsGivenRegister}
    A rational client who correctly runs an instance of Register() also runs correct corresponding instances of CommitClient() and RevealClient().
\end{lemma}

\begin{proof}
    By correctly running Register(), a player deposits $\escrowClient$. The only way to receive $\escrowClient$ back is to correctly run a RevealClient() instance, which itself can only be run after having run a CommitClient() instance in the previous phase. By construction, $\escrowClient$ is greater than any incurrable losses by running CommitClient() and RevealClient(), with maximal losses occurring where the client's order is settled for tokens that have 0 notional (price goes to 0). As the client's initial deposit had notional value strictly less than $\escrowClient$, so must the incurred loss. The result follows.
\end{proof}

\begin{lemma}\label{lem:SINCEMMsGivenCommit}
    A rational MM who correctly runs an instance of CommitMM() also runs a correct instance of RevealMM() in the proceeding phase.
\end{lemma}

\begin{proof}
    By correctly running CommitMM(), a player deposits $\escrowMM$. The only way to receive $\escrowMM$ back is to correctly run a RevealMM() instance in the proceeding. By definition, as the MM's bid and offer has notional value greater than or equal to the total notional in the auction, $\maxAuctionNotional$, so must the incurred loss of not revealing a market. Therefore, players running CommitMM() always correctly run RevealMM().
\end{proof}

\begin{lemma}\label{lem:Clientparticipation}
    Consider an instance of $\protocolName$ between $\tokenA$ and $\tokenB$. A rational client $\playeri$ with $\balanceTokens_i(\bitcoin)>\escrowClient$ runs an instance of Register().
\end{lemma}

\begin{proof}
    Consider such a client $\playeri$ with minimum client utility $\fee$ to exchange one token for another. To execute the swap, $\playeri$ must first call Register(). Given $\playeri$ calls Register(), we know from Lemma \ref{lem:SINCEClientsGivenRegister} that $\playeri$ also calls CommitClient() and RevealClient(). We know from Lemma \ref{lem:ResolutionGetsRun}, Resolution() will be run in every round, meaning $\playeri$ either trades or the tokens are returned. Given a trade occurs, $\playeri$ realises the utility of trading given the restrictions of the order generated by $\playeri$ in CommitClient(), which can be chosen to be any value with positive utility (buy below the MIFP/sell above the MIFP). If no trade occurs, $\playeri$'s realised utility (of 0) does not change. Therefore, $\playeri$ runs Register().
\end{proof}

\begin{lemma}\label{lem:MMparticipation}
    Consider an instance of $\protocolName$ between $\tokenA$ and $\tokenB$, and at least 1 previously called instance of Register(). Any rational MM $\playeri$ with $\balanceTokens_i(\bitcoin) \geq \escrowMM$ and $\balanceTokens_i(\tokenA), $ $\ \balanceTokens_i(\tokenB)$ $>0$ runs an instance of CommitMM().
\end{lemma}

\begin{proof}
    Consider such a player $\playeri$. Given Register() was called by some player $\player_j$, $\playeri$ knows $\player_j$ must call CommitClient() and RevealClient(). Furthermore, $\playeri$ knows the total order size is bounded by $\maxAuctionNotional$ (Section \ref{sec:FBAs}). Given MIFP $\MIFPAB$ and the definition of a MM, there is some market $\market \assign (\MMBidPrice @ \MMOfferPrice)$ at which $\playeri$ observes positive utility to trade with $\player_j$. Therefore, $\playeri$ submits $\market$ to the auction. We must now ensure $\playeri$ submits the order by calling CommitMM(). By the calculation of order settlement, if $\playeri$ submits $\market$ through the necessary Register() and CommitClient() calls and $\playeri$ submits a market order with finite requested width, no trade happens. As such, $\playeri$ runs CommitMM().
\end{proof}

With these lemmas in hand, we have it that rational clients and rational MMs correctly execute all algorithms as outlined by $\protocolName$. This can be expressed concisely in the following corollary. In this corollary, and the theorem that follows, we assume clients and MMs satisfy the token-balance requirements as described in Lemmas \ref{lem:Clientparticipation} and \ref{lem:MMparticipation} respectively. 

\begin{corollary}\label{cor:SINCE}
    Rational clients and MMs always follow the $\protocolName$ protocol.
\end{corollary}

\begin{observation}\label{obs:FTDEXisWSFBA}
    It can be seen that $\protocolName$ with at least $n_\secParam$ previous Register() calls implements a $\ourAuctionName$ when all players follow the protocol. In the Commit phase, CommitClient() specifies a client order which is committed to, while CommitMM() specifies a market which is also committed to. As clients commit to these orders and sign this commitment using a NIZKSoK, nothing is revealed about the client's order, as there are are least $n_\secParam$ Register() calls. This is equivalent to privately submitting the order. 
    
    CommitMM() and RevealMM() ensure MMs provide the equivalent of at least $\maxAuctionNotional$ notional on the bid and offer. Furthermore, as MMs are indistinguishable (see Section \ref{sec:Prelims}), a MM commitment reveals nothing about the liquidity on the bid or offer\footnote{MMs are also allowed to participate as clients if privacy is a concern. CommitMM() provides professionals with a functionality to efficiently provide liquidity in a decentralised setting. It is possible to introduce a RegisterMM() function analogous to Register(), allowing MMs to relay markets in ZK. We believe this has negligible benefit for professionals who already have price and market-size hidden through the commitment scheme.}. 
    
    When the Commit phase is finished, no further orders or markets can be submitted for that auction round, and as such, the clearing price is predetermined. During the Reveal phase, RevealClient() and RevealMM() reveal the orders corresponding to CommitClient() and CommitMM() instances from the previous phase, which are settled in the Resolution phase according to clearing price rules which maximise the amount of notional to be traded, as is in a $\ourAuctionName$.
\end{observation}

\begingroup
\def\thetheorem{\ref{thm:FairTraDEX}}
\begin{theorem}
    Consider an instance of $\protocolName$ between $\tokenA$ and $\tokenB$ with MIFP $\MIFPAB$ and at least $n_\secParam$ previously called instances of Register(). For $ \numMMs$ non-cooperative MMs, the following strategies form strict Nash equilibria:
    \begin{itemize}
        \item $\numMMs=1$: Clients run Register(), followed by CommitClient() producing market orders of width $\fee$. The MM runs CommitMM() producing a market of width at most $ \fee$ with reference price equal to $\MIFPAB$ in size $\maxAuctionNotional$. Clients and MMs then run RevealClient() and RevealMM() respectively.
        \item $\numMMs\geq 2$: 
        Clients run Register(), followed by CommitClient() producing market orders of width greater than 1. MMs run CommitMM() producing markets of width 1 with reference price equal to $\MIFPAB$ in size of at least $\maxAuctionNotional$. Clients and MMs then run RevealClient() and RevealMM() respectively.
    \end{itemize}
\end{theorem}
\endgroup

\begin{proof}
    From Corollary \ref{cor:SINCE}, we know the running of $\protocolName$ is a strictly dominant strategy for clients and MMs. Furthermore, from Observation \ref{obs:FTDEXisWSFBA}, we have seen that $\protocolName$ with at least $n_\secParam$ previously called instances of Register() implements a $\ourAuctionName$. Given $\protocolName$ is a $\ourAuctionName$, the statement follows by applying Theorem \ref{thm:ourFBA}.
\end{proof}

%% file: Appendices/FormalAlgos.tex
\section{\texorpdfstring{$\protocolName$}{TEXT} Algorithms}\label{app:FairTraDEXAlgos}

\begin{figure}
	\includegraphics[width=1\textwidth]{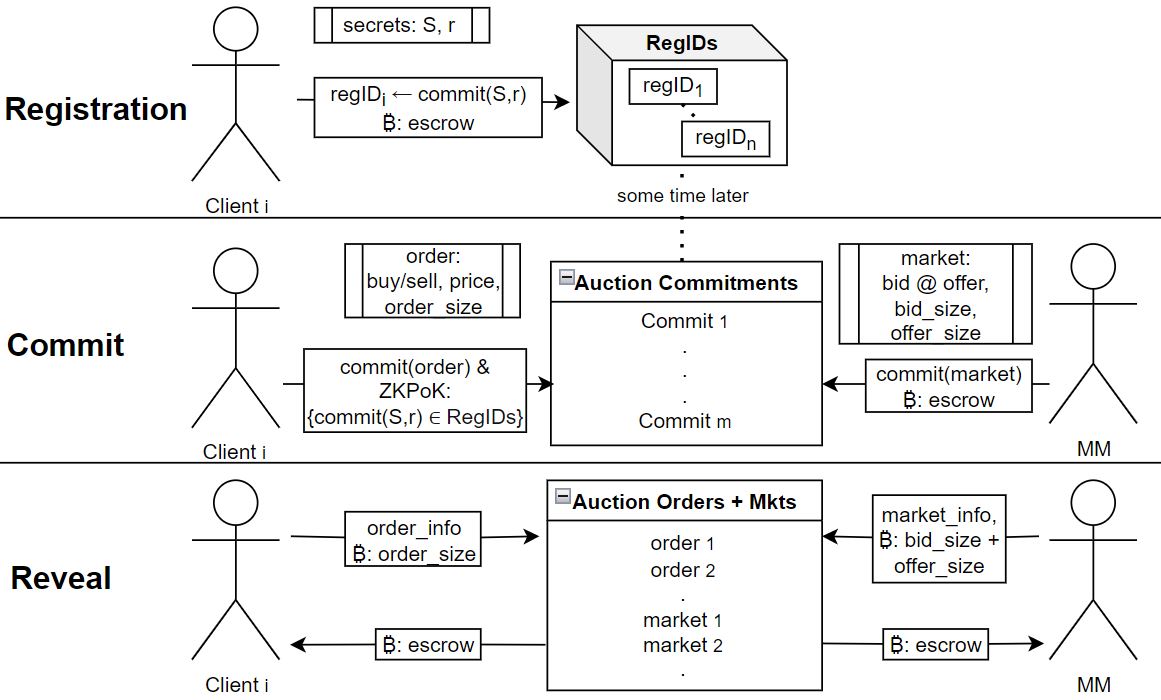}
	\caption{$\protocolName$ phases before order settlement. $\bitcoin:$ indicates the transfer of some tokens, but not necessarily in the same denomination.}
	\label{fig:protocolFlow}
\end{figure}

We now describe in detail the algorithms which together form $\protocolName$, and then describe the main differences between $\protocolName$ and a $\ourAuctionName$.

\begin{itemize}
    \item Setup($\cryptoParam, \maxAuctionNotional ) \rightarrow [\params, \ \MIFPA,\ $ $  \regTokens, \ $ $\MMCommits, $ $ \clientCommits, $ $ \revealedOrders, $ \\ $\revealedMkts]$: For a given cryptographic-security parameter $\cryptoParam$, output the necessary public cryptographic and ZK parameters in $\params$. Set $\MIFPA \in \mathbb{R}^+$ as the indicative price of token $\tokenA$ (used to convert restrictions based on escrows into token amounts). Choose $\escrowClient \in \mathbb{R}^+$ such that the notional of any client order is bounded by $\escrowClient$, and $\escrowMM \assign c \cdot \maxAuctionNotional$, for some $c>1$, with $\maxAuctionNotional$ as described in Section \ref{sec:FBAs}. Choose a bounty $\resBounty \in \mathbb{R}^+$ to reward players for successfully calling Resolution(). Add $\escrowClient, \escrowMM, \maxAuctionNotional, \resBounty$ to $\params$. Set $\regTokens, $ $ \ \clientCommits, $ $ \MMCommits, $ $\revealedOrders, $ $ \revealedMkts   \assign []$.
    
    \item Register($\params,\playeri, \regTokens $) $\rightarrow [( \commSerialNum, \commRandomness),\regToken, $ $\regTokens]$: If $\{ \balanceTokens_i(\bitcoin) $ $\geq \escrowClient \}$, set $\balanceTokens_i(\bitcoin) \assign \balanceTokens_i(\bitcoin)- \escrowClient$. Then, randomly generate $\commSerialNum, \ $ $ \commRandomness \in \{0,1\}^{O(\cryptoParam)}$, and compute $\regToken \assign $ $ \commit(\commSerialNum,\commRandomness)$. Add $\regToken$ to $\regTokens$.
    
    \item CommitClient$(\params, $ $ \regToken,\regTokens, $ $\commSerialNum,$ $ \commRandomness,$ $ \clientCommits,  $ $ \MIFPA) $ $\rightarrow [(\order),$ $\proofZK, $ $\commSerialNum, $ $ \commitment,  $\\  $\clientCommits]$: If $ \{ \regToken \in \regTokens \} $: select a token $\token \in \{\tokenA, \tokenB \}$, a trade price $\tokenPrice \in \mathbb{R}^+ $, and a trade size amount $\tokenAmount \in \mathbb{R}^+$. If token $\token = \tokenA$, set $\MIFP \assign \MIFPA$. Otherwise, set $\MIFP \assign \MIFPA \cdot \tokenPrice $ (used to verify client order is less than the escrow given the indicative price for $\tokenA$).  If $\{ \tokenAmount \leq \balanceTokens_i(\token),  \frac{\escrowClient}{\MIFP} \}$: select a minimum trade width $\width \geq 1$. Set $\order \assign (\token, \tokenAmount,\tokenPrice, \width)$ and $ \commitment \assign \commit(\order)$. Finally, generate the signature of knowledge $\proofZK \assign $ \textit{NIZKSoK}[$\commitment$]$\{(\regToken,$ $\commRandomness):$ MemVerify ($\regTokens,$ $ \regToken )$ $=1 $ $ \And $ $ \regToken=$ $ \commit($ $\commSerialNum, $ $\commRandomness) $ $\}$. If no proof corresponding to $\commSerialNum$ has been computed before, add $\commSerialNum$ to $\clientCommits$. Otherwise, output $\bot$.
    
    \item CommitMM($\params, \playeri, \MMCommits,$ $ \MIFPA) \rightarrow $ $ [(\market),\commitment,  \MMCommits]$ : If $\{ \balanceTokens_i(\bitcoin) \geq \escrowMM \}$: select a bid size $\MMBidAmount \in \mathbb{R}^+$, offer size $\MMOfferAmount\in \mathbb{R}^+$ and prices $0 \leq \MMBidPrice \leq \MMOfferPrice$ such that $\maxAuctionNotional \ \leq \MMBidAmount \cdot \MIFPA \leq \escrowMM, ( \balanceTokens_i(\tokenA)\cdot \MIFPA)$ $\logicalAnd$ $\maxAuctionNotional \ \leq \MMOfferAmount \cdot \MIFPA \cdot \MMOfferPrice \leq \escrowMM, (\balanceTokens_i(\tokenB)\cdot \MIFPA \cdot \MMOfferPrice)$ (ensures the bid, offer prices, and sizes fall within the bounds of the minimum size required by a $\ourAuctionName$ and the escrow). Set $\market \assign (\MMBidPrice, \ \MMBidAmount, \ \MMOfferPrice,  \ \MMOfferAmount)$ and compute $\commitment \assign \commit(\market)$. Finally, set $\balanceTokens_i(\bitcoin) \assign \balanceTokens_i(\bitcoin)- \escrowMM$ and add $\commitment$ to $\MMCommits$. Otherwise, output $\bot$.
    
    \item RevealClient($\params, \proofZK, \commSerialNum, \commRandomness, \order\assign (\token, \tokenAmount,\tokenPrice,$ $ \width, $ $ \MIFPA) $ $ , \clientCommits,$  $ \revealedOrders,$ $ \playeri $ $) \rightarrow $ \\ $[ \revealedOrders]$: If token $\token = \tokenA$, set $\MIFP \assign \MIFPA$. Otherwise, set $\MIFP \assign \MIFPA \cdot \tokenPrice $. If $\{ \proofZK  \in \clientCommits \ \logicalAnd \ \proofZK =$ \textit{NIZKSoK}[$\commit(\order)$]$\{(\regToken, \commRandomness):$ MemVerify($\regTokens, \regToken )=1 \, \And \regToken= \commit(\commSerialNum,\commRandomness) $ $\}  $ $ \logicalAnd \ \tokenAmount \leq \balanceTokens_i(\token), \frac{\escrowClient}{\MIFP}  $ (Repeat checks from CommitClient()): set $\balanceTokens_i(\token) \assign \balanceTokens_i(\token)- \tokenAmount$, $\balanceTokens_i(\bitcoin) \assign \balanceTokens_i(\bitcoin)+ \escrowClient$ (return escrow), and add $\order$ to $\revealedOrders$. Otherwise, output $\bot$.
    
    \item RevealMM($\params, \commitment,\market \assign (\MMBidPrice, \ \MMBidAmount, \ \MMOfferPrice,  \ \MMOfferAmount) , $ $ \MIFPA , $ $\MMCommits, $  \\ $ \revealedMkts,$ $ \playeri) $ $\rightarrow $ $ [ \revealedOrders]$: If $\{\commitment \in $ $ \MMCommits$, $\commitment= $ $\commit(\market) \ $ $ \logicalAnd $ $\maxAuctionNotional \ \leq \MMBidAmount \cdot \MIFPA \leq \escrowMM, ( \balanceTokens_i(\tokenA)\cdot \MIFPA)$ $\logicalAnd$ $\maxAuctionNotional \ \leq \MMOfferAmount \cdot \MIFPA \cdot \MMOfferPrice \leq \escrowMM, (\balanceTokens_i(\tokenB)\cdot \MIFPA \cdot \MMOfferPrice)$ (Repeat checks from CommitMM()): set $\balanceTokens_i(\tokenA) \assign \balanceTokens_i(\tokenA)- \MMBidAmount$, $\balanceTokens_i(\tokenB)$ $ \assign $ $ \balanceTokens_i(\tokenB)- \MMOfferAmount$, $\balanceTokens_i(\bitcoin) \assign \balanceTokens_i(\bitcoin)+\escrowMM$ (return escrow), and add the orders $(\tokenA,\MMBidAmount, \MMBidPrice, \any)$, $(\tokenB,$ $\MMOfferAmount, $ $\MMOfferPrice, \any)$ to $\revealedMkts$. Otherwise, output $\bot$.
    
    \item Resolution($\params, \revealedOrders, \revealedMkts, $ $\clientCommits, $ $ \MMCommits, \playeri) $ $ \rightarrow $ $[\clearingPrice, $ \\ $ \imbalance,$ $ \tightestWidth $ $\clientCommits, $ $\MMCommits,$ $ \revealedOrders, $ $ \revealedMkts ]$: If this is not the first time Resolution() was called, output $\bot$. Otherwise, calculate the tightest market width $\tightestWidth$ of a market in $\revealedMkts$. Remove all markets in $\revealedMkts$ except one with width equal to $\tightestWidth$, chosen using $\revealedMkts$ as a random seed to $\commit()$ \footnote{Given all markets are revealed, the final value of $\revealedMkts$, and as such $\commit(\revealedMkts|| * )$, is unpredictable in the presence of two or more non-cooperative MMs. We prove in Lemma \ref{lem:SINCEMMsGivenCommit} that all MMs running CommitMM() also run RevealMM(). The blockchain based implementation of this function is described in App. \ref{app:tiebreaker}}. Remove all orders in $\revealedOrders$ with requested width greater than $\tightestWidth$. Calculate the clearing price $ \clearingPrice$ which first maximises notional of orders traded from $\revealedOrders \cup \revealedMkts$, and then minimises the imbalance $\imbalance$. 
    
    If $\clearingPrice$ does not maximise notional traded, and then minimise imbalance, output $\bot$. 
    Otherwise, if $\imbalance>0$ at $\clearingPrice$, pro-rate all buy orders with the lowest bid price above $\clearingPrice$ based on order size. If $\imbalance<0$, pro-rate all sell orders with the highest offer price below $\clearingPrice$ based on order size. Then, settle all other buy orders with price greater than or equal to $\clearingPrice$, and all other sell orders with price less than or equal to $\clearingPrice$. Set $ \clientCommits, \MMCommits,$ $ \revealedOrders,  $ $ \revealedMkts \assign []$, and $\balanceTokens_i(\bitcoin) \assign \balanceTokens_i(\bitcoin)+ \resBounty$.
\end{itemize}

\input{Appendices/FairTraDEXvsWSFBAs}

%% file: Appendices/FairTraDEXvsWSFBAs.tex
\subsection{\texorpdfstring{$\protocolName$}{TEXT} vs. \texorpdfstring{$\ourAuctionName$}{TEXT}}\label{app:protocolDiffs}

The main differences between $\protocolName$ and a $\ourAuctionName$ are as follows:
\begin{itemize}
    
    \item Escrows are used to enforce the correct revelation of players who commit to orders or markets. Escrows are only returned to players if orders are revealed and correspond to a valid commit. Furthermore, escrows are chosen large enough to ensure the reclamation of escrows has strictly higher utility than not, ensuring rational players follow the protocol.
    \item An algorithm involving deposits and/or withdrawals updates the set of balances for all players, identifying the player calling the algorithm. 
    \item $\protocolName$ separates the depositing of client escrow  and client order commitments. This is a key functionality necessary to preserve client anonymity and the guarantees of a $\ourAuctionName$. If a client deposits an escrow in the same instance as committing to an order, that information can be used to identify the player, and imply information about the player's order. By separating the two, commitment does not require the update of global variables that can be used to identify the client. 
    \item Set-membership proofs in ZK in the CommitClient() algorithm are used to prove that a player committing to a client order has deposited a client escrow. As $\protocolName$ separates the deposit and commitment steps, these proofs allow a client who deposited an escrow to generate one (and only one, as ZK proofs reveal $\commSerialNum$) order per escrow, while only revealing that the order corresponds to a deposited escrow. As the number of deposited escrows increases, the probability that an order commitment matches any particular escrow approaches 0. This replicates the anonymous order submission of a $\ourAuctionName$.
    \item Tokenised incentives are used to ensure some player in the blockchain calculates the clearing price, and settles orders correctly. 
\end{itemize}

%% file: Appendices/ClearingPriceCalc.tex
\section{Clearing Price Verification} \label{app:clearingPriceCalc}

The protocol in Algorithm \ref{alg:CPVerifier} checks that a given clearing price $\clearingPrice$ clears the highest notional with respect to $\tokenA$. To do this, it checks the imbalance and total notional that would be settled at $\clearingPrice$. Note, that the statements that follow are true given some volume trades at the proposed clearing price, which is asserted in the protocol (line \ref{line:assertVolTrades}).

If the imbalance is 0 (line \ref{line:Imbalance0}), it can be seen that this price must maximise the volume trade while minimising the absolute value of the imbalance. Higher prices reduces the buying notional/ strictly increases the selling notional, while lower prices strictly reduces the selling notional/increases the buying notional, which creates an imbalance without increasing the notional traded. 

If the imbalance is positive (line \ref{line:Imbalance>}), this implies there will be buy orders (partially) unfilled at or above $\clearingPrice$. We can see that the only prices at which more notional can trade must be greater than $\clearingPrice$. Thus, the contract checks the next price increment above (line \ref{line:checkPriceAbove}), and verifies the total notional traded at that price is less than at $\clearingPrice$, or equal, but with a larger absolute imbalance (line \ref{line:assert<}). If the notional traded is lower at this higher price, the clearing price is correct as a lower price reduces the value of the selling notional, and increases the value of the buying notional. 
If the notional traded is the same, but the absolute value of the imbalance is higher, this imbalance must be a sell imbalance by the same reasoning (buying notional decreases, selling notional increases). If the absolute value of the imbalance is higher (although negative), the imbalances at all price points above that price are increasing (buying notional decreases, selling notional increases). 

The same holds for sell imbalances at a proposed clearing price (line \ref{line:Imbalance<}). Checking the price point below (line \ref{line:checkPriceBelow}) and ensuring the notional is lower, or that the notional traded is the same but with a large absolute imbalance (line \ref{line:assert<}) guarantees that the proposed clearing price is valid.

%% file: Appendices/ContractDescription.tex
\section{Encoding of \texorpdfstring{$\protocolName$}{TEXT}}\label{app:Protocol}

The following is an overview of the blockchain-based encoding of $\protocolName$, as described algorithmically in Section \ref{sec:FairTraDEX} and encoded in Appendix \ref{app:ProtocolEncoding}.

For an arbitrary bit-string $m \in \{0,1\}^{*}$, $\relay(m)$ indicates broadcasting $m$ to the relay transaction mempool, where $m$ is included as a transaction in the blockchain if and only if it gives the including relayer access to a relayer fee $\feeRelayer$.

\subsection{Register}

Clients randomly generate $\commSerialNum, \ \commRandomness \in \{0,1\}^{n(\cryptoParam)}$, and compute $\regToken \assign \commit(\commSerialNum,\commRandomness)$. Then, the client sends a $\langle \tagClientRegister,$ $ \regToken \rangle$ to the blockchain using the blockchains PKI. Upon addition to the blockchain, this deposits an escrow of $\escrowClient$ and a relayer fee of $\feeRelayer$ to the blockchain (line \ref{line:depositescrowclient}), with $\regToken$ added to $\clients$ (line \ref{line:addClientReg}).

\subsection{Commit}

A client wishing to submit an order of the form $\order \assign (\token, \tokenAmount,\tokenPrice, \width)$, first generates a commitment to the order $\commitment\assign \commit(\order)$. The client then generates a NIZK signature of knowledge $\proofZK \assign$ NIZKSoK($\commitment$)$\{($ $\regToken, 
$ $\commRandomness):$ MemVerify($\regTokens, \regToken )=1 \, \And \regToken= \commit(\commSerialNum,\commRandomness) $ $\}$ on the commitment. The client then relays a  message of the form $\langle \tagClientCommit, $ $ \commitment, $ $ \commSerialNum, $ $\proofZK \rangle$ to the relayer mempool, which is then sent to the smart contract by a relayer. The transaction is only valid if Verify$(\proofZK)$ returns 1, and as such, a relayer cannot tamper with $\commitment$. The contract first checks that the maximum auction notional has not been reached ($\currentAuctionNotional< \maxAuctionNotional $, line \ref{alg:ClientCommit}).

Furthermore, a valid $\langle \tagClientCommit, \commitment, \commSerialNum, \proofZK \rangle$ message must not reveal a serial number $\commSerialNum$ which has previously been added to $\blacklistedSerialNums $ (initialised line \ref{alg:blackSNsInit}). Serial numbers in $\blacklistedSerialNums$ correspond to client commits that were not revealed according to the protocol during a previous Reveal phase. The escrow corresponding to serial numbers of $\blacklistedSerialNums$ are effectively burned by the protocol, with clients permanently losing access to them. If $\commSerialNum $ is not in $\blacklistedSerialNums$, the order commitment is recorded in $\clientCommits$ (line \ref{alg:addToActiveOrders}), and the relayer who relays the transaction to the blockchain receives the fee (line \ref{alg:payRelayer}).

MMs who wish to participate generate a market of the form $\market \assign (\MMBidPrice, \ \MMBidAmount, \ \MMOfferPrice,  \ \MMOfferAmount)$, and submit a transaction directly to the blockchain of the form $\langle \tagMMCommit, \commit(\market) \rangle$. This transaction deposits an escrow of $\escrowMM$ to the smart contract. The commitment is then recorded in $\MMCommits$ (line \ref{alg:addMMCommit}).

Client and MM Commit transactions are collected until the Commit phase deadline, $\commitDeadline$ (line \ref{alg:deadlineDef}), has passed (line \ref{alg:comDeadlinePassed}).

\subsection{Reveal}\label{app:tiebreaker}

A client who committed to trade through a $\langle \tagClientCommit, \commitment, \commSerialNum, \proofZK \rangle$ transaction in the Commit phase submits a Reveal transaction directly to the blockchain of the form $\langle \tagClientReveal, $ $\token, \tokenAmount, \tokenPrice,  \width, $ $ \commSerialNum, \commRandomness, \regToken$ $ \regTokenNew \rangle$ (line \ref{alg:clientReveal}). If the client intends to ren-enter the protocol as a client, $\regTokenNew$ is a commitment to a new serial number and randomness. Otherwise, it is the null value. 

This $\langle \tagClientReveal, * \rangle$ transaction reveals the token being sold, token amount to sell, and requested width of the order. The $\tokenPrice$ either reveals a limit price at which the client is selling, that the order is the market order if $\tokenPrice= \marketOrder$, or that the client is withdrawing their escrow if $\tokenPrice= \withdraw$. The contract checks:

\begin{itemize}
    \item $\commSerialNum \in \clientCommits $ to verify a commitment corresponding to that serial number has been recorded.
    \item $ \regToken = \commit(\commSerialNum, \commRandomness)$ to ensure that client was indeed the same player that generated the $\regToken$ and that the client order is the same as that committed in the Commit phase $ \commit(\token, \tokenAmount, \tokenPrice, \width)=\clientCommits[\commSerialNum].\commitment $.
\end{itemize}

If the transaction is valid, the order is added to $\revealedBuyOrders$ or $\revealedSellOrders$, depending on direction. If the token being sold by the client is $\tokenA$, the effective order size for clearing price calculation and trade size allocation is the minimum of $\tokenAmount$ and $\escrowClient / \MIFPA$, the maximum token $\tokenA$ order size allowable (line \ref{line:checkAtokensCorrectSize}), with the order recorded in $\revealedOrders$ (line \ref{line:orderAddedA}). 
If the token being sold by the client is $\tokenB$, the order size is the minimum of $\tokenAmount$ and $\escrowClient / (\MIFPA \cdot \MMOfferPrice)$ (line \ref{line:checkBtokensCorrectSize}), with the order recorded in $\revealedOrders$ (line \ref{line:orderAddedB}).

Finally, if $\regTokenNew$ is the null value, the escrow is returned (line \ref{line:returnA} or \ref{line:returnB}), while if it is not, it corresponds to re-entering the protocol as a new client (saving on an additional transaction to re-enter as a client).

A MM who committed to a market through a $\langle \tagMMCommit, \commitment \rangle$ transaction in the commit phase submits $\langle \tagMMReveal, $ $\market \assign$ $(\MMBidPrice, \ \MMBidAmount, \ \MMOfferPrice,  \ \MMOfferAmount) \rangle$ (line \ref{alg:MMReveal}). The contract verifies:

\begin{itemize}
    \item The market $\market$ matches the previously commitment from the commit phase ($\commit($ $\market)$ $= $ $ \MMCommits[\MM].\commitment$) (line \ref{alg:MMReveal}). 
    \item For the bid,  $\maxAuctionNotional / \MIFPA \leq \MMBidAmount$, which verifies the MM has provided the minimum notional required by a $\ourAuctionName$ (line \ref{alg:MMMinNotionalCheck}).
    \item For the offer, $\maxAuctionNotional / (\MIFPA \cdot \MMOfferPrice) \leq \MMOfferAmount$, which again verifies the MM has provided the minimum notional required by a $\ourAuctionName$ (line \ref{alg:MMMinNotionalCheck}).
\end{itemize}

Any MM not revealing a market (line \ref{alg:MMdoesntReveal}) loses their escrow (line \ref{alg:MMnonRevealBurn}) and is prevented from participating in the Resolution phase. Otherwise, from the set of all valid revealed markets $\revealedMkts$, the tightest market is selected, including a tie-breaking procedure for more than one market with width equal to the tightest width. The tie-breaker used in our implementation of $\protocolName$  takes, for a MM $\MM$ (identified by a unique public key) and submitted market $\market$, the market corresponding to the largest value of $\commit( \commit(\revealedMkts) $ $ || \ \MM \ || \ \market)$ (lines \ref{alg:findTightestMMStart}-\ref{alg:findTightestMMEnd}).  Given the tightest market $\market \assign $ $ (\MMBidPrice, \ $ $ \MMBidAmount, \ $ $ \MMOfferPrice,  \ \MMOfferAmount)$ after tie-breaks, the two implicit limit orders are added to the set of revealed orders $\revealedBuyOrders$ and  $\revealedSellOrders$. As was the case with client orders, the effective order size for clearing price calculation and trade size allocation of the bid is the minimum of $\MMBidAmount $ and $\escrowMM / \MIFPA$, while the effective offer size is the minimum of $ \MMOfferAmount$ and $\escrowMM / (\MIFPA \cdot \MMOfferPrice)$.

Reveal transactions are collected until the Reveal deadline, $\revealDeadline$ (line \ref{alg:deadlineDef}), has passed (line \ref{alg:comDeadlinePassed}). 

\subsection{Resolution}

Once the protocol enters the Resolution phase, any player in the system can propose a clearing price by submitting a $\langle \clearingPrice, *  \rangle$ message. Players submitting such a message must deposit a token amount (which we set as $\resBounty$, although any significantly large value to prevent invalid calls to the smart contract will do). This deposit, along with a bounty is returned to the player if $\clearingPrice$ is a valid clearing price.

The orders are then settled based on the clearing price $\clearingPrice$ (lines \ref{line:startOrderSettlement}-\ref{alg:executeOrdersEnd}). If the quantity of $\tokenA$ being sold is greater than the quantity being bought, the sell orders at the highest sell price below the clearing price are pro-rated based on the quantity of $\tokenA$ being sold. If the quantity of token $\tokenA$ being bought is greater than the quantity being sold, the buy orders at the lowest buy price above the clearing price are pro-rated based on the quantity of $\tokenA$ being bought. Remaining unexecuted order balances and escrows are returned to the owners. 

\input{Appendices/ProtocolEncoding/ProtocolEncoding}

%% file: Appendices/ProtocolEncoding/ProtocolEncoding.tex
\subsection{Protocol Encoding}\label{app:ProtocolEncoding}

In the following, for arrays containing array objects, the array objects are uniquely identifiable by the first item in the array (i.e. client identifier, serial number, ZKProof).

\input{Appendices/ProtocolEncoding/Register}

\input{Appendices/ProtocolEncoding/Commit}

\input{Appendices/ProtocolEncoding/Reveal}

\input{Appendices/ProtocolEncoding/Resolution}

%% file: Appendices/ProtocolEncoding/Register.tex
\begin{algorithm}[H]
\def\baselinestretch{1} \scriptsize \raggedright
\caption{Register}
\label{alg:Register} 
	\begin{algorithmic}[1] 
	    \State $\players \assign \textit{generatePopulation}()$
	    \State $\regTokens\assign []$
	    \State $\clientCommits \assign []$
	    \State $\MMCommits \assign []$
	    \State $\phase \assign \null$
	    \State $\tightestWidth \assign \any$
	    \State $\revealedBuyOrders, \ \revealedSellOrders, \revealedMkts \assign []$
	    \State $\lastPhaseChange \assign 0$ \Comment{tracks the block number of last step update}
	    \State $\feeRelayer \assign \textit{getRelayerFee}()$
	    \State $\minTickSize \assign \textit{getMinimumTickSize}()$
	    \State $\commitDeadline, \revealDeadline \assign \revealTXDelay$ \Comment{Deadline for responses equal to the maximal reveal delay described in the Threat Model} \label{alg:deadlineDef}

	    \State $\escrowClient \assign  \escrowMAX$  \Comment{Escrow required to show each market, in line with the Threat Model}
	    
	    \State $\maxAuctionNotional \assign \textit{getMaxAuctionNotional}()$
	    \State $ c \assign \textit{random}(\mathbb{R}_{>1})$
	    \State $\escrowMM \assign  c \cdot \maxAuctionNotional$  \Comment{Escrow required to show each market, some amount greater than $\maxAuctionNotional$}
	    
	    \State $\MIFPA \assign \textit{getTokenAIndicativePrice}()$
	    \State $\currentAuctionNotional \assign 0$
	    \State $\blacklistedSerialNums \assign []$ \Comment{Tracks revealed serial numbers that misbehaved} \label{alg:blackSNsInit}

	    \SPACE 
	
	    \Function{$\textit{Initialise()} $}{}   
		    \State $\phase \assign \phaseCommit$
		\EndFunction

	    \SPACE
	    
	     \Upon{$\langle \tagClientRegister, \regToken \rangle $ $ \from \ \playerProtocol \ \in \ \players $ $ \with \ \playerProtocol\getBalance(\bitcoin) > \escrowClient+\feeRelayer$}{} \Comment{register player as a client} \label{alg:registerClient}
		    \State $\playerProtocol\getTransfer(\escrowClient+\feeRelayer,\bitcoin, \protocolContract)$ \Comment{Add client deposit to the contract account} \label{line:depositescrowclient}
		    \State $\clients\getAppend(\regToken)$ \label{line:addClientReg}
		\EndUpon
		
	\end{algorithmic} 
	
\end{algorithm}

%% file: Appendices/ProtocolEncoding/Commit.tex
\begin{algorithm}[H]
\def\baselinestretch{1} \scriptsize \raggedright
\caption{Commit}
	\begin{algorithmic}[1]
	    
	    \setcounter{ALG@line}{23}
	    
		\Upon{$\relay(\langle \tagClientCommit, \commitment, \commSerialNum, \proofZK \rangle) \ \from \ \playerProtocol  \in  \players  \ \with \ $ $ \currentAuctionNotional< \maxAuctionNotional \logicalAnd \ \text{Verify}(\proofZK, \commitment) =1$ $ \logicalAnd \ \phase = \phaseCommit \ \logicalAnd \ \logicalNot(\commSerialNum \in \blacklistedSerialNums)  $}{} \label{alg:ClientCommit}  
		    \State $\currentAuctionNotional \assign \currentAuctionNotional + \escrowClient$
		    \State $ \clientCommits\getAppend([\commSerialNum, \commitment]) $ \label{alg:addToActiveOrders}
		    \State $\protocolContract\getTransfer(\feeRelayer, \bitcoin, $ $ \playerProtocol))$ \Comment{Reward relayer} \label{alg:payRelayer}
        \EndUpon

	    \SPACE
        
        \Upon{$\langle \tagMMCommit, \commitment \rangle \ \from \ \playerProtocol  \in  \players  $ $ \with  $ $  \playerProtocol\getBalance(\bitcoin) > \escrowMM \logicalAnd $  $ \logicalNot (\playerProtocol \in \MMCommits) \ \phase = \phaseCommit $}{}  \Comment{Allow only one market per player address} \label{alg:MMCommit}
		    \State $\playerProtocol\getTransfer(\escrowMM, \bitcoin, \protocolContract)$ \Comment{Transfer escrow to the protocol contract account}
		    \State $\MMCommits\getAppend([\playerProtocol, \commitment]) $ \label{alg:addMMCommit}
        \EndUpon

	    \SPACE

		\Upon{$ \phase = \phaseCommit \ \logicalAnd \ \blockchain\getHeight() = \lastPhaseChange+\commitDeadline$}{} \label{alg:comDeadlinePassed}    
		    \State $ \phase \assign \phaseReveal$ 
		    \State $ \lastPhaseChange \assign \blockchain\getHeight()$ 
        \EndUpon

	\end{algorithmic} 
	
\end{algorithm}

%% file: Appendices/ProtocolEncoding/Reveal.tex
\begin{algorithm}[H]
\def\baselinestretch{1} \scriptsize \raggedright
\caption{Reveal}
    
	\begin{algorithmic}[1]
	\setcounter{ALG@line}{33}
	
		\Upon{$\langle \tagClientReveal, $ $\token, \tokenAmount, \tokenPrice,  \width, $ $ \commSerialNum, \commRandomness, \regToken$ $, \regTokenNew \rangle \ $ $ \from \ \playerProtocol  \in  \players  $ $  \ \with  \ \commSerialNum \in \clientCommits $ $ \ \logicalAnd \ \regToken = \commit(\commSerialNum, \commRandomness)$ $ \ \logicalAnd \ \commit(\token, \tokenAmount, \tokenPrice=\clientCommits[\commSerialNum].\commitment $ $ \logicalAnd \ \phase =\phaseReveal$}{} \label{alg:clientReveal}
		    \If{$\tokenPrice \neq \textit{withdraw}$}
		    \If{$\token = \tokenA \ \logicalAnd \ \playerProtocol\getBalance(\tokenA)\geq \tokenAmount$ } 
		        \State $\tokenAmount \assign \textit{minimum}(\tokenAmount,  \escrowClient / \MIFPA)$ \label{line:checkAtokensCorrectSize}
		        \State $ \revealedBuyOrders\getAppend([\playerProtocol,  $ $ \tokenAmount,  \tokenPrice, \width ]) $ \Comment{Add client order to array of orders to trade} \label{line:orderAddedA}
		        \State $\playerProtocol\getTransfer(\tokenAmount,\tokenA, \protocolContract)$
		        \If{$ \regTokenNew = \varnothing$} 
		            \State $\protocolContract\getTransfer(\escrowClient, \bitcoin, $ $ \playerProtocol))$  \label{line:returnA}
		        \Else
		            \State $ \clients\getAppend(\regTokenNew)$ 
		        \EndIf
		    \ElsIf{$\token = \tokenB \ \logicalAnd \ \playerProtocol.\getBalance(\tokenB)\geq \tokenAmount$ } 
		        \State $\tokenAmount \assign \textit{minimum}(\tokenAmount,  \escrowClient / (\MIFPA \cdot \tokenPrice))$ \label{line:checkBtokensCorrectSize}
		        \State $ \revealedSellOrders\getAppend([ \playerProtocol,  $ $ \tokenAmount,  \tokenPrice, \width ]) $ \Comment{Add client order to array of orders to trade} \label{line:orderAddedB}
		        
		        \If{$ \regTokenNew = \varnothing$} 
		            \State $\protocolContract\getTransfer(\escrowClient, \bitcoin, $ $ \playerProtocol))$  \label{line:returnB}
		        \ElsIf{$\playerProtocol\getBalance(\bitcoin) > \feeRelayer$}
		            \State $\playerProtocol\getTransfer(\feeRelayer,\bitcoin, \protocolContract)$
		            \State $ \clients\getAppend(\regTokenNew)$ 
		        \EndIf
		    \EndIf
		    \Else \Comment{Client wants to withdraw}
		        \State $\protocolContract\getTransfer(\escrowClient, \bitcoin, $ $ \playerProtocol))$ 
		    \EndIf
		    
		    \State $ \clients\remove(\regToken)$ 
		    \State $ \clientCommits\remove(\commSerialNum)$
		    
        \EndUpon

	    \SPACE
	    
	    \Upon{$\langle \tagMMReveal,\market \assign (\MMBidPrice, \ \MMBidAmount, \ \MMOfferPrice,  \ \MMOfferAmount) \rangle \ \from \ \MM \ \in \MMCommits \ \with \ \commit(\market)= \MMCommits[\MM].\commitment $ $ \logicalAnd \ \phase= \phaseReveal $}{} \label{alg:MMReveal}
	        \If{$ (\maxAuctionNotional / (\MIFPA \cdot \MMOfferPrice) \leq \MMOfferAmount \leq \playerProtocol\getBalance(\tokenB)) $ $ \ \logicalAnd \ (\maxAuctionNotional / \MIFPA  \leq  \MMBidAmount \leq \playerProtocol.\getBalance(\tokenA) )$ } \label{alg:MMMinNotionalCheck} \Comment{Check MM has provided the minimum required liquidity, $\maxAuctionNotional$}
	            \State $\revealedMkts\getAppend(\MM, \market)$
		        \State $\MMCommits\remove(\MM)$
		    \EndIf
        \EndUpon
        
        \SPACE
        
        \Upon{$ \phase = \phaseReveal \ \logicalAnd \  len(\MMCommits)=0 \ \logicalAnd \ len(\clientCommits)=0 $}{} \label{alg:timelyReveal} \Comment{All reveals published}
            \State $ \phase \assign \phaseResolution$ 
		    \State $ \lastPhaseChange \assign \blockchain\getHeight()$
        \EndUpon
        
        \SPACE

		\Upon{$ \phase = \phaseReveal \ \logicalAnd \ \blockchain\getHeight() = \lastPhaseChange+ \revealDeadline $}{} \label{alg:handleNonReveals}
		    \State $\tieBreaker \assign 0$ \label{alg:findTightestMMStart}
	        \State $\tightestMarket \assign ()$
	        \State $\tieBreakSeed \assign \commit(\revealedMkts)$ \Comment{Generate seed for tie-breaks before revealed markets is changed}
		    \For{$\MM \in \revealedMkts$} \Comment{Select the unique market corresponding to the tie-breaker in the proceeding If statement}
		        \If{$ (\maxAuctionNotional / (\MIFPA \cdot \MM.\MMOfferPrice) \leq \MM.\MMOfferAmount \leq \MM\getBalance(\tokenB)) $ $ \ \logicalAnd \ (\maxAuctionNotional / \MIFPA  \leq  \MM.\MMBidAmount \leq \MM\getBalance(\tokenA) )$ } \label{alg:MMMinNotionalCheck2} \Comment{Check MM \textit{still} has provided the minimum required liquidity}
		            \State $\protocolContract\getTransfer(\escrowMM, \bitcoin, $ $ \MM))$ 
		            \If{$\tightestWidth = \any \ \logicalOr \ (\frac{\MMOfferPrice}{\MMBidPrice}<\tightestWidth)$ $\logicalOr \ (\frac{\MMOfferPrice}{\MMBidPrice}= \tightestWidth \logicalAnd \commit(\tieBreakSeed || \MM || \MM.\market ) > \tieBreaker)$}\label{alg:tiebreaker}
	                    \State $\tightestWidth \assign \frac{\MMOfferPrice}{\MMBidPrice}$
	                    \State $\tieBreaker \assign \commit(\tieBreakSeed || \MM || \MM.\market ) $
	                    \State $\tightestMarket \assign [ \MM, \market]$
	                \EndIf
	            \EndIf
	            \State $\revealedMkts\remove(\MM)$
		    \EndFor
		    
		    \State $\MMBidAmount \assign minimum(\tightestMarket.\MMBidAmount,  \escrowMM / \MIFPA)$
		    \State $\MMOfferAmount \assign minimum(\tightestMarket.\MMOfferAmount,  \escrowMM / (\MIFPA\cdot \tightestMarket.\MMOfferPrice))$
		    \State $ \revealedBuyOrders\getAppend([\playerProtocol \assign \tightestMarket.\MM,  $ $ \tokenAmount \assign \MMBidAmount, \tokenPrice \assign \tightestMarket.\MMBidPrice, \width \assign \any ]) $ \Comment{Add tightest market to set of orders to be settled}
		    \State $ \revealedSellOrders\getAppend([\playerProtocol \assign \tightestMarket.\MM,  $ $ \tokenAmount \assign \MMOfferAmount, \tokenPrice \assign  \tightestMarket.\MMOfferPrice, \width \assign  \any ]) $
		    \State $\tightestMarket.\MM\getTransfer(\MMBidAmount,\tokenA, \protocolContract)$
		    \State $\tightestMarket.\MM\getTransfer(\MMOfferAmount,\tokenB, \protocolContract)$ \label{alg:findTightestMMEnd}

		    \For{$\commSerialNum \in \clientCommits$} \Comment{Add all clients who did not reveal order to blacklist, preventing further commitments}
		        \State $\blacklistedSerialNums\getAppend(\commSerialNum)$ \label{alg:blacklistClients}
		        \State $\clientCommits\remove(\commSerialNum)$
		    \EndFor
		    \For{$\MM \in \MMCommits$}\Comment{MMs who did not reveal market in time} \label{alg:MMdoesntReveal}
		        \State $\MMCommits\remove(\MM)$ \Comment{Remove from protocol without adding to $\revealedOrders$, effectively burning escrow} \label{alg:MMnonRevealBurn}
		    \EndFor
		    \State $ \phase \assign \phaseResolution$ 
		    \State $ \lastPhaseChange \assign \blockchain\getHeight()$ 
        \EndUpon

	\end{algorithmic} 
	
\end{algorithm}

%% file: Appendices/ProtocolEncoding/Resolution.tex
\begin{algorithm}[H]
\def\baselinestretch{1} \scriptsize \raggedright
\caption{Resolution: Clearing Price Verification}
\label{alg:CPVerifier} 
	\begin{algorithmic}[1]
	
	\setcounter{ALG@line}{87}

		\Upon{$\langle \clearingPrice, \volumeSettled, \imbalance \rangle $ $\from \playerProtocol \in \players \ \with $ $ \playerProtocol\getBalance(\bitcoin) > \resBounty $ $\logicalAnd \ \phase= \phaseResolution$ }{}
		    \State $\playerProtocol\getTransfer(\resBounty, \bitcoin, \protocolContract)$ \Comment{To prevent Sybil attacks, player must deposit funds which are returned if $\clearingPrice$ is valid} \label{alg:depositResolution}
		    \State $\revealedSellOrders \assign \revealedSellOrders[\revealedSellOrders.\getWidth()>\tightestWidth \ \logicalOr \ \revealedSellOrders\getWidth()=\any]$ \Comment{Remove sell orders that cannot trade due to requested width}
		    \State $\revealedBuyOrders \assign \revealedBuyOrders [\revealedBuyOrders.\getWidth()>\tightestWidth \ \logicalOr \ \revealedBuyOrders\getWidth()=\any]$ 
		    \State \textbf{Assert}$ (\volumeSettled >0 \ \logicalOr \ \text{minimum}(\revealedSellOrders.\tokenPrice)< \text{maximum}(\revealedBuyOrders.\tokenPrice))  $ \label{line:assertVolTrades} \Comment{If the indicated clearing price is below the lowest offer/above highest bid, all of the proceeding checks pass.}

		    \State $\buyVolume \assign \textit{sum}(\revealedBuyOrders[\revealedBuyOrders.\tokenPrice \geq \clearingPrice].\tokenAmount)$
		    \State $\sellVolume \assign \textit{sum}(\revealedSellOrders[\revealedSellOrders.\tokenPrice \leq \clearingPrice].\tokenAmount) $

		    \State \textbf{Assert}(\text{minimum}($\buyVolume / \clearingPrice , \sellVolume)=\volumeSettled $)) 
		    \State \textbf{Assert}($(\buyVolume/ \clearingPrice)-\sellVolume=\imbalance$ ) 
		    \If{$\imbalance = 0$ }\label{line:Imbalance0}\Comment{We are done } 
		            \State $\textit{SettleOrders}(\clearingPrice, \buyVolume, \sellVolume)$ 
		    \EndIf
		        
		    \If{ $\imbalance > 0 $}\label{line:Imbalance>} \Comment{As the auction is bid at $\clearingPrice$, check if next price increment above clears higher volume OR smaller imbalance}
		            \State $\priceToCheck \assign \clearingPrice +\minTickSize$ \label{line:checkPriceAbove}
		            \State $\buyVolumeNew \assign (\buyVolume -\textit{sum}( \revealedBuyOrders[\clearingPrice \leq \revealedBuyOrders.\tokenPrice < \priceToCheck].\tokenAmount)) / \clearingPrice $
		            \State $\sellVolumeNew \assign \sellVolume + \textit{sum}(\revealedSellOrders[\clearingPrice < \revealedSellOrders.\tokenPrice \leq \priceToCheck].\tokenAmount)  $
		            \State \textbf{Assert}($ (\text{minimum} (\buyVolumeNew,\sellVolumeNew)< \volumeSettled) \ \logicalOr $ $( \text{minimum}(\buyVolumeNew,\sellVolumeNew)= \volumeSettled \ \logicalAnd $ $ \ \imbalance \leq |\buyVolumeNew - \sellVolumeNew |)$)\label{line:assert>} \Comment{If the next price clears less volume, or clears the same volume with a larger imbalance, the proposed CP is valid}
		            \State $\textit{SettleOrders}(\clearingPrice, \buyVolume, \sellVolume)$ 
		            
		    \EndIf
		        
		    \If{ $\imbalance < 0$}\label{line:Imbalance<} \Comment{As the auction is offered at $\clearingPrice$, check if next price increment below clears higher volume OR smaller imbalance}
		            \State $\priceToCheck \assign \clearingPrice -\minTickSize$ \label{line:checkPriceBelow}
		            \State $\buyVolumeNew \assign (\buyVolume +\textit{sum}( \revealedBuyOrders[\clearingPrice > \revealedBuyOrders.\tokenPrice \geq \priceToCheck].\tokenAmount) ) / \clearingPrice$
		            \State $\sellVolumeNew \assign \sellVolume - \textit{sum}(\revealedSellOrders[\clearingPrice \geq \revealedSellOrders.\tokenPrice > \priceToCheck].\tokenAmount)$
		            \State \textbf{Assert}($ (\text{minimum} (\buyVolumeNew,\sellVolumeNew)< \volumeSettled) \ \logicalOr $ $( \text{minimum}(\buyVolumeNew,\sellVolumeNew)= \volumeSettled $ $ \logicalAnd \ \imbalance \leq |\buyVolumeNew - \sellVolumeNew |)$)\label{line:assert<}
		            
		            \State $\textit{SettleOrders}(\clearingPrice, \buyVolume, \sellVolume)$

		    \EndIf
		    \State $\protocolContract\getTransfer(2 \resBounty, \bitcoin, $ $ \playerProtocol))$ \Comment{Return deposit, and reward player for submitting valid clearing price} \label{alg:rewardResolution}
		    
        \EndUpon

	\end{algorithmic} 
	
\end{algorithm}	

\begin{algorithm}[H]
\def\baselinestretch{1} \scriptsize \raggedright
\caption{Resolution: Settle Orders}
	\begin{algorithmic}[1]
	
	\setcounter{ALG@line}{111}

		\Function{\textit{SettleOrders}}{$\clearingPrice, \buyVolume, \sellVolume$} \label{alg:tradeResolution}
		    
		    \State $\buyVolume\assign \buyVolume/ \clearingPrice$ \Comment{Convert sell volume to equivalent in $\tokenA$}
		    
		    \If{ $\buyVolume > \sellVolume$} \Comment{pro-rate buy orders at the min price above (or equal to) the clearing price} \label{line:startOrderSettlement}
		        \State $\proRatePrice \assign \textit{minimum}(\revealedBuyOrders[ \revealedBuyOrders.\tokenPrice \geq \clearingPrice].\tokenPrice)$
		        \State $\proRatePriceVolume \assign \textit{sum}(\revealedBuyOrders[\revealedBuyOrders.\tokenPrice =  \proRatePrice].\tokenAmount)/ \clearingPrice$
		        \For{$ \order \in \revealedBuyOrders[\revealedBuyOrders.\tokenPrice =  \proRatePrice] $}
		            \State $ \protocolContract\getTransfer(\order.\tokenAmount \cdot (1-\frac{\buyVolume-\sellVolume}{\proRatePriceVolume}), \tokenA,\order.\playerProtocol)$ \Comment{return tokens not going to be exchanged}
		            \State $\order.\tokenAmount \assign \order.\tokenAmount \cdot \frac{\buyVolume-\sellVolume}{\proRatePriceVolume}$
		        \EndFor
		    \ElsIf{ $\sellVolume > \buyVolume$}  \Comment{ pro-rate sell orders at the max price below (or equal to) the clearing price}
		        \State $\proRatePrice \assign \textit{maximum}(\revealedSellOrders[  \revealedSellOrders.\tokenPrice \leq \clearingPrice].\tokenPrice)$
		        \State $\proRatePriceVolume \assign \textit{sum}(\revealedSellOrders[ \revealedSellOrders.\tokenPrice =  \proRatePrice].\tokenAmount) $
		        \For{$ \order \in \revealedSellOrders[  \revealedSellOrders.\tokenPrice =  \proRatePrice] $}
		            \State $ \protocolContract\getTransfer(\order.\tokenAmount \cdot (1-\frac{\sellVolume-\buyVolume}{\proRatePriceVolume}), \tokenB,\order.\playerProtocol)$ \Comment{return tokens not going to be exchanged}
		            \State $\order.\tokenAmount \assign \order.\tokenAmount \cdot \frac{\sellVolume-\buyVolume}{\proRatePriceVolume}$ \label{line:endOfCPcalculation}
		        \EndFor
		    \EndIf 

	        \SPACE
            
		    \For{$ \order \in \revealedBuyOrders, \revealedSellOrders  $} \Comment{iterate through orders} 
			    \If{$\order \in \revealedBuyOrders \ \logicalAnd \ ( \order.\tokenPrice \geq \clearingPrice \ \logicalOr \ \order.\tokenPrice = \marketOrder)$} \label{alg:executeOrdersBegin} \Comment{execute buy order if bid greater than clearing price} 
				    \State $\tokenTradeSize \assign \order.\tokenAmount / \clearingPrice $
					\State $ \protocolContract\getTransfer(\tokenTradeSize, \tokenB, $ $ \order.\playerProtocol)$
				\ElsIf{$\order \in \revealedSellOrders \ \logicalAnd (\order.\tokenPrice \leq \clearingPrice \  \logicalOr \ \order.\tokenPrice = \marketOrder)$} \label{alg:executeClientSeller} \Comment{execute sell order if bid greater than clearing price}
					\State $\tokenTradeSize \assign (\order.\tokenAmount)/ \clearingPrice $
					\State $ \protocolContract\getTransfer(\tokenTradeSize, \tokenA, $ $ \order.\playerProtocol)$
			
				\Else \Comment{Order not executed}
				    \State $ \protocolContract\getTransfer(\order.\tokenAmount, \order.\token, $ $ \order.\playerProtocol)$
			    \EndIf \label{alg:executeOrdersEnd}
		    \EndFor 
		    
	        \SPACE

	        \State $\phase \assign \phaseCommit$
	         
	        \State $\currentAuctionNotional \assign 0$
	        \State $\revealedBuyOrders, \ \revealedSellOrders \assign []$
	        \State $\tightestWidth \assign \any$
	        \State $\lastPhaseChange \assign \blockchain\getHeight()$
	   \EndFunction

	\end{algorithmic} 
	
\end{algorithm}	

%% file: Appendices/Consideration.tex
\subsection{Existence of irrational players and coalitions}\label{sec:Irrational}

When analysing the optimal strategies of rational players in $\ourAuctionName$s, our results are based on all players being rational and that $n_\secParam$ instances of Register() are called. If we consider the presence of irrational players in the system, we can apply the following adjustments:
\begin{itemize}
    \item \textbf{Irrational MM}: In Lemma \ref{lem:WTTMMstrat}, it is shown that the optimal strategy for a MM is to show markets centred around the MIFP. Any other (irrational) strategy must therefore result in reduced expectancy for the MM, and higher expectancy for clients. Therefore, given the presence of irrational MMs, submitting market orders maximises client expectancy (with greater expected utility than in the presence of rational MMs, although with increased variance).
    \item \textbf{Irrational client}: Given the optimal strategy for rational clients is to submit market orders, a irrational client may then submit limit orders. This merely reduces the irrational client's chance of trading vs. other clients. This would not change the strategy of non-colluding rational MMs, but may have some affect on a monopolistic MM's interpretation of $\fee$.
\end{itemize}

Furthermore, if less than $n_\secParam$ instances of Register() are called, clients resort to submitting limit orders. This can be seen in the proof of Theorem \ref{thm:ourFBA}. In the proof, if clients are not sure that a MM will show a market with reference price equal to the MIFP, the case when less than $n_\secParam$ instances of Register() are called, the optimal strategy for clients is to submit limit orders, which only stands to reduce the clients' probability of trading. As the number of non-cooperative players in $\protocolName$ decreases towards two, the guarantees of $\protocolName$ approach those of an AMM. However, as client price and order size remain hidden until the counterparty chooses her strategy, and before the clearing price is fixed (end of the Commit phase), $\protocolName$ maintains advantages over AMMs against client-based EEV attacks, such as price/order-size specific front-running and selective participation. 

%% file: Appendices/NotesOnFTDEX.tex
\section{Practical Considerations for \texorpdfstring{$\protocolName$}{TEXT}}\label{app:discussion}

Practical approaches to ensure equivalent to $n_\secParam$ Register() calls are taken in existing anonymity protocols. For example, Tornado Cash \footnote{\url{https://torn.community/t/anonymity-mining-technical-overview/15} Accessed: 20/07/2022} rewards players proportionally to the (Tornado Cash equivalent of the) number of Register() calls, as well as the length of time between calling Register() and CommitClient(). ,

\subsection*{Escrow choices}

Choosing escrow amounts for clients and MMs should reflect the emergent use cases of the protocol. It is possible to create different $\protocolName$ instances for the same pair of tokens based on trade size, both for liquidity purposes (MMs will require wider markets for larger-sized orders, but the corresponding increased escrow requirements might prevent smaller clients from participating) and auction use-cases (day-trading vs. end-of-day balancing). Furthermore, as the escrow denomination ($\bitcoin$) in our description is different to at least one of the tokens, there needs to be some way to translate the escrow amount into order sizes. This will depend on the environment, but on Ethereum for example,  price oracles (existing AMMs, Chainlink\footnote{https://chain.link/}, etc. ) can be used. It is also possible to use previous clearing prices from within the $\protocolName$ ecosystem, although a self-referential oracle must be implemented carefully as there may be game-theoretic implications. If a satisfactory price oracle exists, deposits can be made in the respective tokens of the swap, further reducing the capital requirements for players and encouraging adoption.

\subsection*{Incentive compatibility given transaction fees}

In Section \ref{sec:FairTraDEX}, we mention that our Nash equilibria are dependent on the utility gained by clients and MMs being greater than the cost for participation. The choice of smart-contract enabled blockchain on which to deploy will dictate the barrier of entry for clients and MMs alike. Like existing attempts to implement blockchain-based FBAs \cite{PubliclyVerifiableSecrecyPreservingPeriodicAuctionsGalal,BlockAuctionPeriodicAuctionsConstantinides}, we have an amortised number of transactions per player of two. A naive comparison to AMMs, where this is reduced to 1 for clients, and 0 for MMs, certainly has less direct costs than $\protocolName$. However, when factors like impermanent loss, slippage\footnote{\url{https://docs.uniswap.org/protocol/concepts/V3-overview/swaps}}, front-running, and EEV attacks in general, the value being extracted from DEXs incurs a significant indirect cost for clients. Immediately, we can increase the expected cost of using AMMs by the slippage required by AMMs (set to $0.5\%$ as of writing in Uniswap V3, but for larger orders this must increase by the nature of AMMs). We can increase this further by the probability orders are not executed (where prices move more than the slippage, potentially due to front-running) but are added to the blockchain. As such, the indirect costs are substantial, are increasing in order-size increases and proportionally to improvements in client trading ability as strategies can be replicated/front-run. A thorough comparison of the monetary costs of $\protocolName$ vs. AMMs over various order-sizes, and trading scenarios makes for interesting future work as $\protocolName$ begins to be deployed and tested in the wild. 

It can be seen from the proof of Theorem \ref{thm:ourFBA} that the client strategies identified are strong incentive compatible in expectation as the MM always shows markets of width less than or equal to $\fee$. However, the MM strategy of showing width 1 markets is not strong incentive compatible. In addition to the fees described in the protocol, an additional fee can be applied within the protocol itself to incentivise the participation of MMs. This can be a function of MM participation/market widths. As with all additional fees/rewards, the game-theoretic implications of such an incentivisation scheme must be considered. 

In our encoding of $\protocolName$, we do not explicitly introduce a cost for clients and MMs in Commit/Reveal phases to reward the submission of the clearing price. In reality, the result in Lemma \ref{lem:ResolutionGetsRun} holds without the introduction of an explicit reward, as there all participating clients and MMs will have positive expectancy to receive tokens through correct order resolution. The use of an explicit reward is for illustrative purposes, and to avoid complications regarding transaction fees for running the clearing price checks. The costs of running the Resolution contracts must be ensured to be less than the utility gained by at least one player in the blockchain protocol for calling the contract. 

\subsection*{Differences to previous versions of $\protocolName$}

In a previous version of this paper, the notional amount $\maxAuctionNotional$ was was also used to upper-bound the notional of client orders allowable in an auction. This was intended to simplify the strategy analysis in the proof of Theorem \ref{thm:ourFBA}, although without adding any additional guarantees to the protocol. In the current version of the paper, we manage to provide a simpler definition of a $\ourAuctionName$ without significantly adding to the complexity of Theorem \ref{thm:ourFBA} by noting any MM order through the MIFP must be negative in expectancy for the MM.